\documentclass[11pt]{article}
\usepackage{amsmath,amsthm}
\usepackage{amssymb, amsfonts}
\usepackage{graphics,epsfig, natbib,bbm,bm}
\usepackage[ansinew]{inputenc}
\usepackage[commentmarkup=footnote]{changes}
\usepackage{lipsum}
\definechangesauthor[name={Soeren}, color=blue]{soe}
\definechangesauthor[name={Luis}, color=red]{lu}
\usepackage{tikz,pgfplots}
\usetikzlibrary{arrows}

\newtheorem{theorem}{Theorem}[section]

\newtheorem{lemma}[theorem]{Lemma}
\newtheorem{proposition}[theorem]{Proposition}

\newenvironment{remark}{\addtocounter{theorem}{1}\vskip 0.2cm{\sc Remark
\thetheorem.}}{\hfill\vskip 0.2cm}

\newcommand{\argmax}{\operatornamewithlimits{argmax}}

\newcommand{\y}{w}

\DeclareMathOperator{\sgn}{sgn}

\numberwithin{equation}{section}

\begin{document}

\title{A Solvable Two-dimensional Optimal Stopping Problem in the Presence of Ambiguity}

\author{Luis H. R. Alvarez E.\thanks{Department of Accounting and Finance, Turku School of Economics,
FIN-20014 University of Turku, Finland, E-mail: lhralv@utu.fi}\quad
Sören Christensen\thanks{Mathematisches Seminar,
Christian-Albrechts-Universität zu Kiel,
Ludewig-Meyn-Str. 4,
D-24098 Kiel,
Germany,  E-mail: christensen@math.uni-kiel.de}}
\maketitle

\abstract{According to conventional wisdom, ambiguity accelerates optimal timing by decreasing the value of waiting in comparison with the unambiguous benchmark case. We study this mechanism in a  multidimensional setting and show that in a multifactor model ambiguity does not only influence the rate at which the underlying processes are expected to grow, it also affects the rate at which the problem is discounted. This mechanism where nature also selects the rate at which the problem is discounted cannot appear in a one-dimensional setting and as such we identify an indirect way of how ambiguity affects optimal timing.}\\



\thispagestyle{empty} \clearpage \setcounter{page}{1}

\section{Introduction}

Ambiguity is known to accelerate the rational exercise of timing decisions in comparison with the unambiguous benchmark model.
The main reason for this finding is that higher ambiguity decreases the value associated with the worst case scenario while leaving the exercise payoff unchanged. Since waiting is, however, optimal as long as the value of waiting dominates the exercise payoff, ambiguity tends to shrink the continuation set where waiting is optimal and, consequently, speeds up optimal timing. Since ambiguity interplays with the volatility of the underlying stochastic factor dynamics it is clear that its impact on the timing decision is more involved and not necessarily symmetric in a multidimensional setting. Our objective is to investigate this question in a class of solvable two-dimensional stopping problems.

The literature investigating ambiguity and its impact on decision making is extensive.  Within an atemporal multiple priors setting the theory originates from the path breaking study \citet{GiSch89}
(see also \citet{Be02}, \citet{Kli_et_al_05},
\citet{Ma_et_al_06} and \citet{NiOz06}). The axiomatization based on atemproal analysis was subsequently extended into an intertemporal
recursive multiple priors setting by, among others, \cite{EpWa94}, \citet{ChEp02}, \cite{EpMi03}, and \citet{EpSch03}. The impact of ambiguity on optimal timing decisions was first investigated by
\cite{NiOz04} in a job search model. The analysis of ambiguity on optimal timing decisions have subsequently been extended
to various directions. \cite{NiOz07} considered the impact of Knightian uncertainty on
the optimal investment timing decisions in a continuous time model based on geometric Brownian motion. \cite{Al07} investigated the impact of Knightian uncertainty on monotone one-sided stopping problems and expressed the value as well as the optimality conditions for the stopping boundaries in terms of the monotone fundamental solutions characterizing the Green-function of the underlying diffusion process.
\cite{Ri2009} developed within a discrete time setting
a general minmax martingale approach to optimal stopping problems in the presence of ambiguity aversion.  \cite{ChRi2013}, in turn, extended the analysis developed in \cite{Ri2009} to a continuous time setting and identified
the value as the smallest right continuous $g$-martingale dominating the payoff process.
\cite{MiWa_11} considered the impact of ambiguity on timing in a model based on a general discrete time Feller-continuous Markov process.
\cite{Chr13} investigates optimal stopping of linear diffusions in the presence of Knightian uncertainty and identifies explicitly the minimal excessive mappings generating the worst case measure and states a general characterization for the value and optimal timing policy in terms of these mappings.

We consider the impact of ambiguity on optimal timing in a multidimensional setting by solving an optimal stopping problem in the case where the exercise payoff is positively homogeneous and the underlying diffusions are geometric Brownian motions. Since a positively homogeneous function is not necessarily continuous, our results cast simultaneously light on the optimal policy and its value also in discontinuous multidimensional cases. By utilizing the fact that the ratio of two geometric Brownian motions constitutes a geometric Brownian motion, we first reduce the dimensionality of the considered problem and express it in terms of an associated one-dimensional problem which can be analyzed by relying on the approach based on minimal excessive functions developed in \cite{Chr13}. Even though the transformed problem is one-dimensional, it differs in two significant ways from standard one dimensional timing problems in the presence of ambiguity. First of all, in the transformed problem ambiguity does not only affect the rate at which the value is expected to grow, it also affects the rate at which the exercise payoff is discounted. This mechanism where nature also selects the rate at which the problem is discounted cannot appear in an ordinary one-dimensional setting. Therefore, our results present a new indirect way of how ambiguity affects optimal timing. Second, since the optimal timing decision is affected by two random factors the density generators characterizing the worst case measure may separately switch from one extreme to another at different interconnected states. This is again a result which cannot arise in a single factor setting.

The contents of this study are as follows. In section 2 we describe the underlying stochastic dynamics and state the considered problem. Our main findings are then summarized in section 3. Our general findings are illustrated in various explicit examples in section 4. Section 5 then concludes our study.

\section{Underlying Dynamics and Problem Setting}
Let $\mathbf{W}_{t}=(W_{1t},W_{2t})^T$ be an ordinary two-dimensional Brownian motion under the measure $\mathbb{P}$ and assume that the underlying processes follow under the measure $\mathbb{P}$ the stochastic dynamics characterized by the stochastic differential equations
\begin{align}
dX_t &= \mu_X X_t dt + \sigma_X X_t dW_{1t},\quad X_0=x\in \mathbb{R},\\
dY_t &= \mu_Y Y_t dt + \sigma_Y Y_t dW_{2t},\quad Y_0=y\in \mathbb{R},
\end{align}
where $\mu_X,\mu_Y\in \mathbb{R}$ and  $\sigma_X,\sigma_Y\in \mathbb{R}_{+}$ are known constants.

As usually in models subject to Knightian uncertainty, let the degree of ambiguity $\kappa>0$ be given and denote by $\mathcal{P}^\kappa$ the set of all probability measures, that are equivalent to $\mathbb{P}$ with density process of the form
$$
\mathcal{M}_t^{{\bm \theta}}=e^{-\int_0^t {\bm \theta}_s d\mathbf{W}_s - \frac{1}{2}\int_0^t \|{\bm \theta_s}\|^2 ds}
$$
for a progressively measurable process $\{{\bm \theta}_t\}_{t\geq 0}$ with $|\theta_{it}|\leq \kappa$ for all $t\geq 0$ and $i=1,2$. Under the measure $\mathbb{Q}^{{\bm \theta}}$ defined by the likelihood ratio
$$
\frac{d\mathbb{Q}^{{\bm \theta}}}{d\mathbb{P}}=\mathcal{M}_t^{{\bm \theta}}
$$
we naturally have that
\begin{align*}
\tilde{W}_{1t}^{\theta_1} &= W_{1t} + \int_0^t\theta_{1s}ds\\
\tilde{W}_{2t}^{\theta_2} &= W_{2t}+ \int_0^t\theta_{2s}ds
\end{align*}
is an ordinary 2-dimensional $\mathbb{Q}^{{\bm \theta}}$-Brownian motion. Thus, we notice that under a measure $\mathbb{Q}^{{\bm \theta}}\in \mathcal{P}^\kappa$ the dynamics of the underlying processes read as
\begin{align}
dX_t &= (\mu_X-\sigma_X\theta_{1t}) X_t dt + \sigma_X X_t d\tilde{W}_{1t}^{\theta_1},\quad X_0=x\in \mathbb{R},\label{X}\\
dY_t &= (\mu_Y -\sigma_Y \theta_{2t})Y_t dt + \sigma_Y Y_t d\tilde{W}_{2t}^{\theta_2},\quad Y_0=y\in \mathbb{R},\label{Y}
\end{align}
where $\mathbf{\tilde{W}}_{t}=(\tilde{W}_{1t}^{\theta_1},\tilde{W}_{2t}^{\theta_2})^T$ denotes a two-dimensional $\mathbb{Q}^{{\bm \theta}}$-Brownian motion.

Given the underlying processes and the class of equivalent measures generated by the density process $\mathcal{M}_t^{{\bm \theta}}$, we now plan to investigate the following optimal stopping problem
\begin{align}
V_\kappa(x,y) = \sup_{\tau\in \mathcal{T}}\inf_{\mathbb{Q}^{{\bm \theta}}\in \mathcal{P}^\kappa}\mathbb{E}_{\mathbf{x}}^{{\mathbb{Q}^{{\bm \theta}}}}\left[e^{-r\tau}F(X_\tau,Y_\tau)\mathbbm{1}_{\tau < \infty}\right],\label{stopping}
\end{align}
where $F:\mathbb{R}_+^2\mapsto \mathbb{R}$ is a known non-negative and measurable function which is assumed to be positively homogeneous of degree one in the following unless otherwise stated. For the sake of comparison, denote the value of the optimal timing policy in the absence of ambiguity by (cf. \cite{AlVi05} and \cite{ChIr11})
\begin{align}
V_0(x,y) = \sup_{\tau\in \mathcal{T}}\mathbb{E}_{\mathbf{x}}^{\mathbb{P}}\left[e^{-r\tau}F(X_\tau,Y_\tau)\mathbbm{1}_{\tau < \infty}\right].\label{stoppingnoamb}
\end{align}
Before stating our main results on the optimal timing policy and its value we first establish the following result characterizing the impact of ambiguity on the optimal stopping strategy and its value in a general setting.
\begin{lemma}\label{comp}
Ambiguity decreases the value of the optimal stopping policy and accelerates optimal timing by shrinking the continuation region where waiting is optimal. More precisely, for a general measurable and non-negative reward function $F$ it holds that
$V_\kappa(x,y)\leq V_0(x,y)$ for all $(x,y)\in \mathbb{R}_+^2$ and $C_\kappa\subseteq C_0$, where $C_\kappa=\{(x,y)\in \mathbb{R}_+^2:V_\kappa(x,y)>F(x,y)\}$ and $C_0=\{(x,y)\in \mathbb{R}_+^2:V_0(x,y)>F(x,y)\}$.
\end{lemma}
\begin{proof}
Inequality $V_\kappa(x,y)\leq V_0(x,y)$ follows directly from the definition of the value of the optimal stopping policy in the presence of ambiguity. Denote the continuation regions associated to the considered stopping problems by $C_\kappa=\{(x,y)\in \mathbb{R}_+^2:V_\kappa(x,y)>F(x,y)\}$ and $C_0=\{(x,y)\in \mathbb{R}_+^2:V_0(x,y)>F(x,y)\}$. It is clear that if $(x,y)\in C_\kappa$ then
$V_0(x,y)\geq V_\kappa(x,y)>F(x,y)$ implying that $(x,y)\in C_0$ as well and, consequently, that $C_\kappa\subseteq C_0$.
\end{proof}
Lemma \ref{comp} demonstrates that ambiguity accelerates optimal exercise by shrinking the continuation set at which waiting is optimal. As intuitively is clear higher ambiguity also decreases the value of the optimal policy. It is worth emphasizing that the negativity of the impact of ambiguity on the value and the incentives to wait is more generally valid than just within the considered class of problems, since the proof can be directly extended to a higher-dimensional setting where the underlying processes are more general than just geometric Brownian motions.

\section{Optimal Timing Policy}

Our objective is now to develop our main results on the considered class of optimal stopping problems. We can now make the following useful observation  summarizing a set of conditions under which the worst case prior can be straightforwardly identified and, consequently, under which both the value as well as the optimal stopping policy can be determined from a standard stopping problem.
\begin{lemma}\label{monotonic}
For a general measurable and non-negative reward function $F$ the following  holds true:\\
(A) Assume that $F(x,y)$ is monotonically increasing as a function of $x$ and monotonically decreasing as a function of $y$. Then
the worst case measure $\mathbb{Q}^{\ast}$ is generated by the choice $(\theta_{1t}^\ast,\theta_{2t}^\ast)=(\kappa,-\kappa)$.\\
(B) Assume that $F(x,y)$ is monotonically increasing as a function of $x$ and $y$. Then
the worst case measure $\mathbb{Q}^{\ast}$ is generated by the choice $(\theta_{1t}^\ast,\theta_{2t}^\ast)=(\kappa,\kappa)$.\\
(C) Assume that $F(x,y)$ is monotonically decreasing as a function of $x$ and monotonically increasing as a function of $y$. Then
the worst case measure $\mathbb{Q}^{\ast}$ is generated by the choice $(\theta_{1t}^\ast,\theta_{2t}^\ast)=(-\kappa,\kappa)$.\\
(D) Assume that $F(x,y)$ is monotonically decreasing as a function of $x$ and $y$. Then
the worst case measure $\mathbb{Q}^{\ast}$ is generated by the choice $(\theta_{1t}^\ast,\theta_{2t}^\ast)=(-\kappa,-\kappa)$.\\
In all cases stated above we have
\begin{align*}
V_\kappa(x,y) = \sup_{\tau\in \mathcal{T}}\mathbb{E}_{\mathbf{x}}^{{\mathbb{Q}^{{\bm \theta}^\ast}}}\left[e^{-r\tau}F(X_\tau,Y_\tau)\mathbbm{1}_{\tau<\infty}\right]
\end{align*}
\end{lemma}
\begin{proof}
We only prove part (A), since proving the rest of the claims is completely analogous. It is now clear by definition of the value $V_\kappa(x,y)$ that
\begin{align*}
V_\kappa(x,y) \leq \sup_{\tau\in \mathcal{T}}\mathbb{E}_{\mathbf{x}}^{{\mathbb{Q}^{(\kappa,-\kappa)}}}\left[e^{-r\tau}F(X_\tau,Y_\tau)\mathbbm{1}_{\tau<\infty}\right].
\end{align*}
In order to reverse this inequality, we notice that since
$$
F\left(X_t,Y_t\right)=F\left(xe^{\int_0^t\left(\mu_X-\frac{1}{2}\sigma_X^2-\sigma_X\theta_{1s}\right)ds + \sigma_X \tilde{W}_{1t}^{\theta_1}},ye^{\int_0^t\left(\mu_Y-\frac{1}{2}\sigma_Y^2-\sigma_Y\theta_{2s}\right)ds + \sigma_Y \tilde{W}_{2t}^{\theta_2}}\right)
$$
the alleged claim follows directly from standard comparison results after invoking the a.s.-bounds
\begin{align*}
X_t &\geq  xe^{\left(\mu_X-\frac{1}{2}\sigma_X^2-\sigma_X\kappa\right)t + \sigma_X \tilde{W}_{1t}^{\theta_1}}\\
Y_t &\leq  ye^{\left(\mu_Y-\frac{1}{2}\sigma_Y^2+\sigma_Y\kappa\right)t + \sigma_Y \tilde{W}_{2t}^{\theta_2}},
\end{align*}
the strong uniqueness of the solutions, and the assumed monotonicity of the exercise payoff.
\end{proof}
Lemma \ref{monotonic} characterizes the worst case measures in the case where the reward function is strictly monotonic as a function of the underlying state variables.
This simplifies the analysis since it delineates circumstances under which the worst case measure can be determined solely based on the monotonicity properties of the reward
without having to solve simultaneously the worst case density generators and the value of the optimal policy. Since the measure is in the cases treated in Lemma \ref{monotonic} independent of the prevailing state, we notice that
under the conditions of Lemma \ref{monotonic} the value of the optimal policy preserves the homogeneity of the exercise payoff (cf. \cite{OlSt92} for a general treatment in the absence of ambiguity; see also \cite{McSi86} and \cite{HuOk98}).

Having stated the cases associated with the monotone cases, we now plan to proceed into the analysis of the more general cases resulting into endogenous state dependent switching of the density generators determining the worst case measure. In order to characterize the worst case measure in the considered class of problems, let us first consider the determination of the optimal density generators by relying on standard dynamic programming arguments. To this end, we denote by
$$
\mathcal{A}^{\theta_1,\theta_2}=\frac{1}{2}\sigma_X^2x^2\frac{\partial^2}{\partial x^2}+\frac{1}{2}\sigma_Y^2y^2\frac{\partial^2}{\partial y^2}+(\mu_X-\sigma_X\theta_1)x\frac{\partial}{\partial x}+(\mu_Y-\sigma_Y\theta_2)y\frac{\partial}{\partial y}
$$
the differential operator associated with the underlying processes $(X_t,Y_t)$ under the measure $\mathbb{Q}^{{\bm \theta}}\in \mathcal{P}^\kappa$ and
assume now that the function $u:\mathbb{R}_+^2\mapsto \mathbb{R}_+$ is twice continuously differentiable on $\mathbb{R}_+^2$.
A standard application of the It{\^o}-Döblin theorem then yields
\begin{align}\label{id}
\begin{split}
e^{-rt}u(X_{t},Y_{t})&=u(x,y) + \int_0^{t}e^{-rs}\left((\mathcal{A}^{\theta_1,\theta_2}u)(X_s,Y_s)-ru(X_s,Y_s)\right)ds\\
&+\int_0^{t}e^{-rs}\left(u_x(X_s,Y_s)\sigma_X X_sd\tilde{W}_{1s}^{\theta_1} +
u_y(X_s,Y_s)\sigma_Y Y_sd\tilde{W}_{2s}^{\theta_2} \right).
\end{split}
\end{align}
It is then clear from this expression that the density generators associated with the worst case scenario are of the form
$\theta_{1t}^\ast=\kappa \sgn(u_x(X_t,Y_t))$
and
$\theta_{2t}^\ast=\kappa \sgn(u_y(X_t,Y_t)).$ Thus, if the function $u(x,y)$ is chosen so that it satisfies the partial differential equation
$(\mathcal{A}^{\theta_1^\ast,\theta_2^\ast}u)(x,y)-ru(x,y)=0$ on some open subset $G$ with compact closure in $\mathbb{R}_+^2$, then we have
\begin{align}\label{id2}
\begin{split}
e^{-rt}u(X_{t},Y_{t})&=\int_0^{t}e^{-rs}\left((\theta_{1s}^\ast-\theta_{1s})\sigma_XX_su_x(X_s,Y_s)+
(\theta_{2s}^\ast-\theta_{2s})\sigma_YY_su_y(X_s,Y_s)\right)ds\\
&+u(x,y) + \int_0^{t}e^{-rs}\left(u_x(X_s,Y_s)\sigma_X X_sd\tilde{W}_{1s}^{\theta_1} +
u_y(X_s,Y_s)\sigma_Y Y_sd\tilde{W}_{2s}^{\theta_2} \right)
\end{split}
\end{align}
for all admissible density generators $(\theta_1,\theta_2)$, $(x,y)\in G$, and $t\leq \tau_G=\inf\{t\geq 0:(X_t,Y_t)\not \in G\}$.
Noticing that $(\theta_{1}^\ast-\theta_{1})u_x(x,y)\geq 0$ and $(\theta_{2}^\ast-\theta_{2})u_y(x,y)\geq 0$ for all admissible
$(\theta_1,\theta_2)$ and $(x,y)\in G$ then shows that
$$
\mathbb{E}^{\mathbb{Q}^{{\bm \theta}}}_{\mathbf{x}}\left[e^{-rt\wedge \tau_G}u(X_{t\wedge \tau_G},Y_{t\wedge \tau_G})\right] \geq u(x,y)
$$
with equality only when $(\theta_1,\theta_2)=(\kappa \sgn(u_x(x,y)),\kappa \sgn(u_y(x,y)))$. Even though this observation is interesting as a
characterization for the optimal policy and its value, it has at least two weaknesses from the perspective of the considered class of problems.
First of all, it is not beforehand clear whether the value of the optimal policy is actually smooth enough for the utilization of the
It{\^o}-Döblin theorem. Fortunately, there are extensions which do not require as much smoothness (see, for example, pp. 315 -- 318 in \cite{Oksendal2003} or Section IV.7 in \cite{Pr05}) which could be utilized in the determination of the value and worst case measure. Second, the characterization of the value on the continuation region as
a solution of a partial differential equation overlooks the homogeneity properties of the considered class of problems and, therefore, does not make use of the possibility to reduce the dimensionality of the considered problem by focusing on the ratio process $Z_t:=X_t/Y_t$ instead of the two-dimensional process $(X_t,Y_t)$.

Given the observations above, let us now follow the approach originally developed in \cite{Chr13} and investigate if we can identify a set of minimal harmonic functions which can be utilized in the determination of the value of an optimal timing policy. We also refer to \cite{BeLe1997} (see also \cite{LeUr2007}, \cite{ChIr11}, and \cite{GaLe2011}) for related considerations for usual optimal stopping problems.

Put formally, our objective is to identify a twice continuously differentiable function $h_c:\mathbb{R}_+\mapsto \mathbb{R}_+$ such that $u(x,y)=yh_c(z)$, where $z=x/y$. In order to achieve this we plan to utilize the positive homogeneity of the exercise payoff and investigate if it is possible to find homogeneous solutions for the partial differential equation characterizing the value on the continuation region. In the present setting
\begin{align*}
\begin{split}
(\mathcal{A}^{\theta_1,\theta_2}u)(x,y)-ru(x,y) &= y\left(\frac{1}{2}(\sigma_X^2+\sigma_Y^2)z^2h_c''(c)+(\mu_X-\mu_Y)zh_c'(z)-(r-\mu_Y)h_c(z)\right)\\
&+y\left(\sigma_Y\theta_2(zh_c'(z)-h_c(z))-zh_c'(z)\sigma_X\theta_1\right)
\end{split}
\end{align*}
indicating that the density generators resulting into a worst case measure are now
$\theta_{1}^\ast=\kappa \sgn(h_c'(z))$ and $\theta_{2}^\ast=
\kappa\sgn(h_c(z)-zh_c'(z))$. Consequently, if $(\mathcal{A}^{\theta_1^\ast,\theta_2^\ast}u)(x,y)=ru(x,y)$ then
we necessarily have
$$
\frac{1}{2}(\sigma_X^2+\sigma_Y^2)z^2h_c''(c)+(\mu_X-\mu_Y+\kappa(\sigma_X+\sigma_Y))zh_c'(z)-(r-\mu_Y+\sigma_Y \kappa)h_c(z) = 0
$$
for $z\in A_1=\{z\in\mathbb{R}_+:h_c'(z)\leq 0\}$,
$$
\frac{1}{2}(\sigma_X^2+\sigma_Y^2)z^2h_c''(c)+(\mu_X-\mu_Y-\kappa(\sigma_X-\sigma_Y))zh_c'(z)-(r-\mu_Y+\sigma_Y \kappa)h_c(z) = 0
$$
for $z\in A_2=\{z\in\mathbb{R}_+: 0 < zh_c'(z) \leq h_c(z)\}$, and
$$
\frac{1}{2}(\sigma_X^2+\sigma_Y^2)z^2h_c''(c)+(\mu_X-\mu_Y-\kappa(\sigma_X+\sigma_Y))zh_c'(z)-(r-\mu_Y-\sigma_Y \kappa)h_c(z) = 0
$$
for $x\in A_3=\{z\in\mathbb{R}_+: zh_c'(z) > h_c(z)\}$. Moreover, $(\theta_1^\ast,\theta_2^\ast)=(-\kappa,\kappa)$ for $z\in A_1$, $(\theta_1^\ast,\theta_2^\ast)=(\kappa,\kappa)$ for $z\in A_2$, and $(\theta_1^\ast,\theta_2^\ast)=(\kappa,-\kappa)$ for $z\in A_3$.

In order to determine the harmonic mappings needed for the determination of the value and its optimal policy, we first notice that if condition $r>\mu_Y-\kappa\sigma_Y$ is satisfied then
the quadratic equation
\begin{align}
\frac{1}{2}(\sigma_X^2+\sigma_Y^2)a(a-1)+(\mu_X-\mu_Y+\kappa(\sigma_X+\sigma_Y))a+\mu_Y-\kappa\sigma_Y-r=0\label{quad1}
\end{align}
has two roots
\begin{align*}
\psi_{-\kappa}&=\frac{1}{2}-\frac{\mu_X-\mu_Y+\kappa(\sigma_X+\sigma_Y)}{\sigma_X^2+\sigma_Y^2}+
\sqrt{\left(\frac{1}{2}-\frac{\mu_X-\mu_Y+\kappa(\sigma_X+\sigma_Y)}{\sigma_X^2+\sigma_Y^2}\right)^2+
\frac{2(r-\mu_Y+\kappa\sigma_Y)}{\sigma_X^2+\sigma_Y^2}}>0\\
\varphi_{-\kappa}&=\frac{1}{2}-\frac{\mu_X-\mu_Y+\kappa(\sigma_X+\sigma_Y)}{\sigma_X^2+\sigma_Y^2}-
\sqrt{\left(\frac{1}{2}-\frac{\mu_X-\mu_Y+\kappa(\sigma_X+\sigma_Y)}{\sigma_X^2+\sigma_Y^2}\right)^2+
\frac{2(r-\mu_Y+\kappa\sigma_Y)}{\sigma_X^2+\sigma_Y^2}}<0.
\end{align*}
In that case we can define the monotonically decreasing, twice continuously differentiable, and strictly convex function
$h_\infty(z):=z^{\varphi_{-\kappa}}.$ On the other hand, condition $r>\mu_Y-\kappa\sigma_Y$ also guarantees that
the quadratic equation
\begin{align}
\frac{1}{2}(\sigma_X^2+\sigma_Y^2)a(a-1)+(\mu_X-\mu_Y-\kappa(\sigma_X-\sigma_Y))a+\mu_Y-\kappa\sigma_Y-r=0\label{quad2}
\end{align}
has two roots
\begin{align*}
\hat{\psi}_\kappa&=\frac{1}{2}-\frac{\mu_X-\mu_Y-\kappa(\sigma_X-\sigma_Y)}{\sigma_X^2+\sigma_Y^2}+
\sqrt{\left(\frac{1}{2}-\frac{\mu_X-\mu_Y-\kappa(\sigma_X-\sigma_Y)}{\sigma_X^2+\sigma_Y^2}\right)^2+
\frac{2(r-\mu_Y+\kappa\sigma_Y)}{\sigma_X^2+\sigma_Y^2}}>0\\
\hat{\varphi}_{\kappa}&=\frac{1}{2}-\frac{\mu_X-\mu_Y-\kappa(\sigma_X-\sigma_Y)}{\sigma_X^2+\sigma_Y^2}-
\sqrt{\left(\frac{1}{2}-\frac{\mu_X-\mu_Y-\kappa(\sigma_X-\sigma_Y)}{\sigma_X^2+\sigma_Y^2}\right)^2+
\frac{2(r-\mu_Y+\kappa\sigma_Y)}{\sigma_X^2+\sigma_Y^2}}<0.
\end{align*}
In this case we can define the monotonically increasing and twice continuously differentiable function $h_0(z):=z^{\hat{\psi}_{\kappa}}$ when $r\leq \mu_X-\kappa\sigma_X$ and
$h_0(z):=z^{\psi_{\kappa}}$ when $r> \mu_X-\kappa\sigma_X$. Notice that condition $r> \mu_X-\kappa\sigma_X$ implies that $\psi_{\kappa}>1$ guaranteeing that $h_0(z)$ is strictly convex in that case. Finally, in order to capture the potential cases appearing in multiple boundary problems, we again assume that condition $r>\mu_Y-\kappa\sigma_Y$ is satisfied and define the twice continuously differentiable function $h_c:\mathbb{R}_+\mapsto\mathbb{R}_+$ as
\begin{align}\label{harmonica}
h_c(z)=\begin{cases}
\frac{\hat{\psi}_\kappa}{\hat{\psi}_\kappa-\hat{\varphi}_\kappa}\left(\frac{z}{c}\right)^{\hat{\varphi}_\kappa}-
\frac{\hat{\varphi}_\kappa}{\hat{\psi}_\kappa-\hat{\varphi}_\kappa}\left(\frac{z}{c}\right)^{\hat{\psi}_\kappa},&z\geq c,\\
\frac{\psi_{-\kappa}}{\psi_{-\kappa}-\varphi_{-\kappa}}\left(\frac{z}{c}\right)^{\varphi_{-\kappa}}-
\frac{\varphi_{-\kappa}}{\psi_{-\kappa}-\varphi_{-\kappa}}\left(\frac{z}{c}\right)^{\psi_{-\kappa}},&z\leq c,
\end{cases}
\end{align}
when $r\leq \mu_X-\kappa\sigma_X$. If, however, $r> \mu_X-\kappa\sigma_X$ then
\begin{align}\label{harmonic}
h_c(z)=\begin{cases}
\frac{\hat{\psi}_\kappa}{\hat{\psi}_\kappa-1}l^{\hat{\varphi}_\kappa}\left(\frac{1-\varphi_\kappa}{\psi_\kappa-\varphi_\kappa}\left(\frac{z}{l c}\right)^{\psi_\kappa} + \frac{\psi_\kappa-1}{\psi_\kappa-\varphi_\kappa}\left(\frac{z}{l c}\right)^{\varphi_\kappa}\right), &z\geq lc,\\
\frac{\hat{\psi}_\kappa}{\hat{\psi}_\kappa-\hat{\varphi}_\kappa}\left(\frac{z}{c}\right)^{\hat{\varphi}_\kappa}-
\frac{\hat{\varphi}_\kappa}{\hat{\psi}_\kappa-\hat{\varphi}_\kappa}\left(\frac{z}{c}\right)^{\hat{\psi}_\kappa},&c\leq z\leq lc,\\
\frac{\psi_{-\kappa}}{\psi_{-\kappa}-\varphi_{-\kappa}}\left(\frac{z}{c}\right)^{\varphi_{-\kappa}}-
\frac{\varphi_{-\kappa}}{\psi_{-\kappa}-\varphi_{-\kappa}}\left(\frac{z}{c}\right)^{\psi_{-\kappa}},&z\leq c,
\end{cases}
\end{align}
where
$$
l=\left(\frac{\hat{\psi}_\kappa(1-\hat{\varphi}_\kappa)}
{\hat{\varphi}_\kappa(1-\hat{\psi}_\kappa)}\right)^{\frac{1}{\hat{\psi}_\kappa-\hat{\varphi}_\kappa}}.
$$
Since
\begin{align*}
(lc)^2h_c''(lc+) &=\frac{\hat{\psi}_\kappa}{\hat{\psi}_\kappa-1}l^{\hat{\varphi}_\kappa}(\psi_\kappa-1)(1-\varphi_\kappa)=
\frac{\hat{\psi}_\kappa}{\hat{\psi}_\kappa-1}l^{\hat{\varphi}_\kappa}\frac{2(r-\mu_Y+\kappa\sigma_X)}{\sigma_X^2+\sigma_Y^2}\\
(lc)^2h_c''(lc-) &=\frac{\hat{\psi}_\kappa}{\hat{\psi}_\kappa-1}l^{\hat{\varphi}_\kappa}(\hat{\psi}_\kappa-1)(1-\hat{\varphi}_\kappa)=
\frac{\hat{\psi}_\kappa}{\hat{\psi}_\kappa-1}l^{\hat{\varphi}_\kappa}\frac{2(r-\mu_Y+\kappa\sigma_X)}{\sigma_X^2+\sigma_Y^2}
\end{align*}
we notice that $h_c(z)$ is again twice continuously differentiable on $\mathbb{R}_+$,
monotonically decreasing on $(0,c)$, and monotonically increasing on $(c,\infty)$. Moreover, condition $r> \mu_X-\kappa\sigma_X$ guarantees that $h_c(z)$ is strictly convex on $\mathbb{R}_+$ as well, and satisfies the condition $h_c'(z)z>h_c(z)$ on $(lc,\infty)$.
It is worth noticing that if $r>\mu_X-\kappa\sigma_X$, then the function $h_c(z)$ admits the representation
$h_c(z)=\max(H_{1c}(z),H_{2c}(z),H_{3c}(z))$, where
\begin{align*}
H_{1c}(z) &=\frac{\psi_{-\kappa}}{\psi_{-\kappa}-\varphi_{-\kappa}}\left(\frac{z}{c}\right)^{\varphi_{-\kappa}}-
\frac{\varphi_{-\kappa}}{\psi_{-\kappa}-\varphi_{-\kappa}}\left(\frac{z}{c}\right)^{\psi_{-\kappa}},\\
H_{2c}(z) &=
\frac{\hat{\psi}_\kappa}{\hat{\psi}_\kappa-\hat{\varphi}_\kappa}\left(\frac{z}{c}\right)^{\hat{\varphi}_\kappa}-
\frac{\hat{\varphi}_\kappa}{\hat{\psi}_\kappa-\hat{\varphi}_\kappa}\left(\frac{z}{c}\right)^{\hat{\psi}_\kappa},\\
H_{3c}(z) &=\frac{\hat{\psi}_\kappa}{\hat{\psi}_\kappa-1}l^{\hat{\varphi}_\kappa}\left(\frac{1-\varphi_\kappa}{\psi_\kappa-\varphi_\kappa}\left(\frac{z}{l c}\right)^{\psi_\kappa} + \frac{\psi_\kappa-1}{\psi_\kappa-\varphi_\kappa}\left(\frac{z}{l c}\right)^{\varphi_\kappa}\right).
\end{align*}
Note that since $l\uparrow\infty$ as $r\rightarrow \mu_X-\kappa\sigma_X$ the function $h_{c}(z)$ reduces to $h_c(z)=\max(H_{1c}(z),H_{2c}(z))$
when $\mu_Y-\kappa\sigma_Y<r<\mu_X-\kappa\sigma_X$.

Having characterized the key harmonic functions needed in the characterization of the value of the optimal timing policy, we now investigate the behavior of the ratio
$$
\Pi_c(z) = \frac{F(z,1)}{h_c(z)}
$$
for all $z\in \mathbb{R}_+$ and $c\in[0,\infty]$. A set of auxiliary results characterizing the key role of the harmonic functions in the determination of the worst case measure as well as the value of the optimal policy is now summarized in the following (see Lemma 1 in \cite{Chr13}).

\begin{lemma}\label{l1}
Assume that $r>\max(\mu_X-\kappa\sigma_X,\mu_Y-\kappa\sigma_Y)$ and denote by $\mathbb{Q}^{{\bm \theta}^c}\in \mathcal{P}^\kappa$ the measure induced by the density generators $\bm{\theta}_t^c:=\kappa(\sgn(X_t-cY_t),\sgn(lc Y_t-X_t))^T$, where $c\in[0,\infty]$. Then,
\begin{align}
dX_t&=(\mu_X-\kappa\sigma_X\sgn(X_t-cY_t))X_tdt+\sigma_X X_t d\tilde{W}_t^{\theta_1^c},\label{sde1}\\
dY_t&=(\mu_Y+\kappa\sigma_Y\sgn(X_t-lcY_t))Y_tdt+\sigma_Y Y_t d\tilde{W}_t^{\theta_2^c},\label{sde2}
\end{align}
and
\begin{align}\label{ratio}
\begin{split}
dZ_t &=(\mu_X-\mu_Y+\sigma_Y^2-\kappa\sigma_X\sgn(Z_t-c)-\kappa\sigma_Y \sgn(Z_t-lc))Z_tdt\\
&+\sigma_X Z_t d\tilde{W}_{1t}^{\theta_1^c}-\sigma_Y Z_t d\tilde{W}_{2t}^{\theta_2^c},
\end{split}
\end{align}
where $(\tilde{W}_{1t}^{\theta_1^c},\tilde{W}_{2t}^{\theta_2^c})$ is a 2-dimensional Brownian motion under the measure $\mathbb{Q}^{{\bm \theta}^c}$. Moreover, for any stopping time $\tau\in \mathcal{T}$, admissible density generator $\bm{\theta}$, and $(x,y)\in \mathbb{R}_+^2$ we have
$$
\mathbb{E}_{\mathbf{x}}^{\mathbb{Q}^{{\bm \theta}^c}}\left[e^{-r\tau}Y_\tau h_{c}(X_\tau/Y_\tau)\mathbbm{1}_{\{\tau<\infty\}}\right]\leq y h_c(x/y)\leq \mathbb{E}^{\mathbb{Q}^{{\bm \theta}}}_{\mathbf{x}}\left[e^{-r\tau}Y_\tau h_{c}(X_\tau/Y_\tau)\mathbbm{1}_{\{\tau<\infty\}}\right].
$$
\end{lemma}
\begin{proof}
We first observe that under our assumptions $h_c(z)$ is strictly convex and twice continuously differentiable on $\mathbb{R}_+$. Consequently, the standard It{\^o}-Döblin theorem applies.
Given this observation and utilizing the identity \eqref{id2} in the homogeneous case $u(x,y)=yh_c(z)$ yields
\begin{align*}
\begin{split}
e^{-rt}Y_{t}h_c(Z_t)&=\int_0^{t}e^{-rs}Y_s\left((\kappa\sgn(h_c'(Z_s))-\theta_{1s})\sigma_XZ_sh_c'(Z_s)+(\kappa\sgn(\Delta(Z_s))-
\theta_{2s})\sigma_Y\Delta(Z_s)\right)ds\\
&+ yh_c(z)+\int_0^{t}e^{-rs}Y_s\left(\sigma_X Z_sh_c'(Z_s)d\tilde{W}_{1s}^{\theta_1} +
\sigma_Y\Delta(Z_s) d\tilde{W}_{2s}^{\theta_2} \right)\\
&=\int_0^{t}e^{-rs}Y_s\left((\kappa\sgn(Z_s-c)-\theta_{1s})\sigma_XZ_sh_c'(Z_s)+(\kappa\sgn(lc-Z_s)-
\theta_{2s})\sigma_Y\Delta(Z_s)\right)ds\\
&+ yh_c(z)+\int_0^{t}e^{-rs}Y_s\left(\sigma_X Z_sh_c'(Z_s)d\tilde{W}_{1s}^{\theta_1} +
\sigma_Y\Delta(Z_s) d\tilde{W}_{2s}^{\theta_2} \right),
\end{split}
\end{align*}
where $\Delta(z)=h_c(z)-h_c'(z)z$. Assume that $G\subset\mathbb{R}_+^2$ is an open subset with compact closure in $\mathbb{R}_+^2$ and let $\tau_G=\inf\{t\geq0: (X_t,Y_t)\not\in G\}$ denote the first exit time of the process $(X_t,Y_t)$ from $G$. The smoothness of the function $yh_c(z)$ and behavior of the ratio process $Z_t$ (it can be sandwiched between two geometric Brownian motions) then guarantees that all the functional forms are bounded on open subsets with compact closure on $\mathbb{R}_+^2$. Since $Y_t(\kappa\sgn(Z_t-c)-\theta_{1t})Z_th_c'(Z_t)\geq 0$ and $Y_t(\kappa\sgn(lc-Z_t)-
\theta_{2t})\Delta(Z_t)\geq 0$ for all $t\geq 0$ and admissible density generators $(\theta_{1t},\theta_{2t})$ we notice that
$$
e^{-rt}Y_{t}h_c(Z_t) \geq yh_c(z)+\int_0^{t}e^{-rs}Y_s\left(\sigma_X Z_sh_c'(Z_s)d\tilde{W}_{1s}^{\theta_1} +
\sigma_Y\Delta(Z_s) d\tilde{W}_{2s}^{\theta_2} \right)
$$
demonstrating that the stopped process
$\{e^{-r(t\wedge \tau_G)}Y_{t\wedge \tau_G}h_c(Z_{t\wedge \tau_G})\}_{t\geq 0}$ is a bounded positive $\mathbb{Q}^{{\bm \theta}}$-submartingale. For $\bm{\theta}_t = \bm{\theta}_t^c:=\kappa(\sgn(Z_t-c),\sgn(lc-Z_t))^T$ we have
\begin{align*}
\begin{split}
e^{-rt}Y_{t}h_c(Z_t)&=yh_c(z)+\int_0^{t}e^{-rs}Y_s\left(\sigma_X Z_sh_c'(Z_s)d\tilde{W}_{1s}^{\theta_1^c} +
\sigma_Y\Delta(Z_s) d\tilde{W}_{2s}^{\theta_2^c} \right),
\end{split}
\end{align*}
proving that the process $\{e^{-r(t\wedge \tau_G)}Y_{t\wedge \tau_G}h_c(Z_{t\wedge \tau_G})\}_{t\geq 0}$ is a bounded positive local $\mathbb{Q}^{{\bm \theta}^c}$-martingale. Analogous computations demonstrate that the process $\{e^{-rt}Y_{t}h_c(Z_{t})\}_{t\geq 0}$ is actually a positive $\mathbb{Q}^{{\bm \theta}^c}$-martingale and, therefore, a supermartingale. Finally, it is clear that under the measure $\mathbb{Q}^{{\bm \theta}^c}$ the underlying processes as well as their ratio evolve according to the random dynamics characterized by the stochastic differential equations \eqref{sde1}, \eqref{sde2}, and \eqref{ratio}.
\end{proof}
Lemma \ref{l1} essentially shows how the function $h_c(z)$ induces an appropriate class of worst case supermartingales for the considered class of processes. A first result of this type is given in the next {proposition}.

\begin{proposition}\label{prop:max_pt_stopping}
 If
$$
z^\ast\in \argmax\left\{\Pi_c(z)\right\},
$$
then $\{(x,y)\in \mathbb{R}_+^2: x=z^\ast y\}\subseteq\Gamma_\kappa:=\{(x,y)\in \mathbb{R}_+^2: V_\kappa(x,y)=F(x,y)\}$.
\end{proposition}

\begin{proof}
	Assume that the set $\argmax\left\{\Pi_c(z)\right\}\neq \emptyset$ and let
	$
	z^\ast\in \argmax\left\{\Pi_c(z)\right\}.
	$
	It is then clear that (see part (i) of Theorem 3 in \cite{Chr13})
	\begin{align*}
	\inf_{\mathbb{Q}^{{\bm \theta}}\in \mathcal{P}^\kappa}\mathbb{E}_{\mathbf{x}}^{{\mathbb{Q}^{{\bm \theta}}}}\left[e^{-r\tau}F(X_\tau,Y_\tau)\mathbbm{1}_{\{\tau<\infty\}}\right]&=\inf_{\mathbb{Q}^{{\bm \theta}}\in \mathcal{P}^\kappa}\mathbb{E}_{\mathbf{x}}^{{\mathbb{Q}^{{\bm \theta}}}}\left[e^{-r\tau}Y_\tau h_c(Z_\tau)\Pi_c(Z_\tau)\mathbbm{1}_{\{\tau<\infty\}}\right]\\
	&\leq \Pi_c(z^\ast)\inf_{\mathbb{Q}^{{\bm \theta}}\in \mathcal{P}^\kappa}\mathbb{E}_{\mathbf{x}}^{{\mathbb{Q}^{{\bm \theta}}}}\left[e^{-r\tau}Y_\tau h_c(Z_\tau)\mathbbm{1}_{\{\tau<\infty\}}\right]\\
	&= \Pi_c(z^\ast)\mathbb{E}_{\mathbf{x}}^{{\mathbb{Q}^{{\bm \theta}^c}}}\left[e^{-r\tau}Y_\tau h_c(Z_\tau)\mathbbm{1}_{\{\tau<\infty\}}\right]\\
	&\leq \Pi_c(z^\ast)yh_c(z)
	\end{align*}
	for all $\tau\in \mathcal{T}$, $(x,y)\in\mathbb{R}_+^2$, and $c\in \mathbb{R}_+$. Hence, we have found that $V_\kappa(x,y) \leq \Pi_c(z^\ast)yh_c(z)$.
	Since $V_\kappa(x,y) \geq F(x,y)$ for all $(x,y)\in \mathbb{R}_+^2$ we notice that
	$$
	\Pi_c(z)yh_c(z)\leq V_\kappa(x,y) \leq \Pi_c(z^\ast)yh_c(z)
	$$
	proving that $\{(x,y)\in \mathbb{R}_+^2: x=z^\ast y\}\subseteq\Gamma_\kappa:=\{(x,y)\in \mathbb{R}_+^2: V_\kappa(x,y)=F(x,y)\}$ as claimed.
	
\end{proof}

Interestingly, the
ratio process \eqref{ratio} is an ordinary linear diffusion with known infinitesimal characteristics which helps us in the determination of the expected first exit times from bounded open intervals in
$\mathbb{R}_+$. This is formally summarized in our next lemma.
\begin{lemma}\label{finiteness}
Assume that $(a,b)\subset \mathbb{R}_+$ so that $0< a < b <\infty$. For all $z\in (a,b)$ we have
$$
\mathbb{E}_{z}^{{\mathbb{Q}^{{\bm \theta}^{c}}}}\left[\inf\{t\geq0:Z_t\not\in(a,b)\}\right]=\int_a^b G_0(z,y)m_\kappa(y)dy < \infty,
$$
where
$$
G_0(x,y) = \begin{cases}
B_0^{-1} (S_\kappa(y)-S_\kappa(a))(S_\kappa(b)-S_\kappa(x)),&x\geq y,\\
B_0^{-1} (S_\kappa(y)-S_\kappa(a))(S_\kappa(b)-S_\kappa(x)),&x\leq y,
\end{cases}
$$
$B_0=S_\kappa(b)-S_\kappa(a)$,
$$
S_\kappa'(z)=\begin{cases}
c^{-\frac{2(\mu_1-\mu_3)}{\Sigma^2}}l^{-\frac{2(\mu_2-\mu_3)}{\Sigma^2}} z^{-\frac{2\mu_3}{\Sigma^2}}, & z\in[lc,\infty),\\
c^{-\frac{2(\mu_1-\mu_2)}{\Sigma^2}}z^{-\frac{2\mu_2}{\Sigma^2}}, & z\in [c,lc),\\
z^{-\frac{2\mu_1}{\Sigma^2}}, & z\in(0,c),
\end{cases}
$$
$
m_\kappa'(z)=2/(\Sigma^2 z^2 S_\kappa'(z)),
$
$\Sigma^2=\sigma_X^2+\sigma_Y^2$, $\mu_1=\mu_X-\mu_Y+\sigma_Y^2+\kappa(\sigma_X+\sigma_Y)$, $\mu_2=\mu_X-\mu_Y+\sigma_Y^2-\kappa(\sigma_X-\sigma_Y)$, and $\mu_3=\mu_X-\mu_Y+\sigma_Y^2-\kappa(\sigma_X+\sigma_Y)$.
\end{lemma}
\begin{proof}
$S_\kappa'(z)$ constitutes the density of the scale function and $m_\kappa'(z)$ the density of the speed measure of the diffusion $Z$ characterized by
\eqref{ratio}. The result then follows from the analysis in, for example, Chapter 4 in \cite{ItoMcK1974}.
\end{proof}

Lemma \ref{finiteness} essentially guarantees that the first exit times of $Z$ from bounded open intervals in
$\mathbb{R}_+$ is $\mathbb{Q}^{{\bm \theta}^{c}}$-a.s. finite. This observation as well as the properties of the function $h_c$ plays a central role in the characterization of the optimal timing policy and its value as summarized in our next theorem.
\begin{theorem}[Two-sided case]\label{t1}
	Assume that there are two points $z_i^\ast\in \argmax\left\{\Pi_{c^\ast}(z)\right\}$, $i=1,2$, for some $c^\ast\in \mathbb{R}_+$ such that $\Pi_{c^\ast}(z_1^\ast)=\Pi_{c^\ast}(z_2^\ast)$. Then
\begin{itemize}
  \item[(A)]
$V_\kappa(x,y) =
\Pi_{c^\ast}(z_i^\ast)yh_{c^\ast}(x/y) \mbox{ for all $(x,y)\in\mathbb{R}_+^2$ with }x/y\in (z_1^\ast,z_2^\ast)\mbox{ and $i=1,2$}.
$
  \item[(B)] If $\Pi_{c^\ast}(z_i^\ast)>\Pi_{c^\ast}(z)$ for all $z\in (z_1^\ast,z_2^\ast)$, then $$\{(x,y)\in \mathbb{R}_+^2: z_1^\ast y < x < z_2^\ast y\}\subseteq C_\kappa:=\{(x,y)\in \mathbb{R}_+^2:V_\kappa(x,y)>F(x,y)\}.$$
  \item[(C)] Let $\tau^\ast\in \mathcal T$ be such that  $\tau^\ast=\inf\{t\geq 0:Z_t\not\in (z_1^\ast,z_2^\ast)\}$ $\mathbb P_{\bf x}$-a.s. for all initial points $(x,y)$ with $z_1^\ast y < x < z_2^\ast y$. Then, $(\mathbb{Q}^{{\bm \theta}^{c^\ast}},\tau^\ast)$ is an equilibrium in the sense that for all initial points $(x,y)$ with $z_1^\ast y < x < z_2^\ast y$ it holds that
  \begin{align*}
  \mathbb{E}_{\mathbf{x}}^{{\mathbb{Q}^{{\bm \theta}}}}\left[e^{-r\tau^{\ast}}F(X_{\tau^{\ast}},Y_{\tau^{\ast}})\mathbbm{1}_{\{\tau^{\ast}<\infty\}}\right]&\geq \mathbb{E}_{\mathbf{x}}^{\mathbb{Q}^{{\bm \theta}^{c^\ast}}}\left[e^{-r\tau^{\ast}}F(X_{\tau^{\ast}},Y_{\tau^{\ast}})\mathbbm{1}_{\{\tau^{\ast}<\infty\}}\right]\mbox{ for all }\mathbb{Q}^{{\bm \theta}}\in \mathcal{P}^\kappa\\
  \mathbb{E}_{\mathbf{x}}^{\mathbb{Q}^{{\bm \theta}^{c^\ast}}}\left[e^{-r\tau}F(X_{\tau},Y_{\tau})\mathbbm{1}_{\{\tau<\infty\}}\right]&\leq \mathbb{E}_{\mathbf{x}}^{\mathbb{Q}^{{\bm \theta}^{c^\ast}}}\left[e^{-r\tau^{\ast}}F(X_{\tau^{\ast}},Y_{\tau^{\ast}})\mathbbm{1}_{\{\tau^{\ast}<\infty\}}\right]\mbox{ for all }\tau\in\mathcal T.
  \end{align*}
\end{itemize}
\end{theorem}
\begin{proof}
	First note that  the process
	$$
	M_t=e^{-rt}Y_th_{c^\ast}(X_t/Y_t)
	$$
	is a positive $\mathbb{Q}^{{\bm \theta}^{c^\ast}}$-martingale by Lemma \ref{l1}. Using this, it holds that for all $\tau\in\mathcal T$ and all initial points $(x,y)$ with $z_1^\ast y < x < z_2^\ast y$
	\begin{align*}
	\mathbb{E}_{\mathbf{x}}^{\mathbb{Q}^{{\bm \theta}^{c^\ast}}}\left[e^{-r\tau}F(X_{\tau},Y_{\tau})\mathbbm{1}_{\{\tau<\infty\}}\right]&=\mathbb{E}_{\mathbf{x}}^{\mathbb{Q}^{{\bm \theta}^{c^\ast}}}\left[M_\tau \Pi_{c^\ast}(Z_\tau)\mathbbm{1}_{\{\tau<\infty\}}\right]\\
	&\leq \sup_z \Pi_{c^\ast}(z)\mathbb{E}_{\mathbf{x}}^{\mathbb{Q}^{{\bm \theta}^{c^\ast}}}\left[M_\tau \mathbbm{1}_{\{\tau<\infty\}}\right]\\
	&\leq\sup_z \Pi_{c^\ast}(z)\mathbb{E}_{\mathbf{x}}^{\mathbb{Q}^{{\bm \theta}^{c^\ast}}}\left[M_0\right]\\
	&=\Pi_{c^\ast}(z_i)yh_{c^\ast}(x/y)
	\end{align*}
	It furthermore holds that $\tau^\ast$ as given in (C) is $\mathbb{Q}^{{\bm \theta}^{c^\ast}}$-a.s. finite by Lemma \ref{finiteness} for all initial points $(x,y)$ with $z_1^\ast y < x < z_2^\ast y$
	and hence 
	 $Z_{\tau^*}\in\{z_1,z_2\}$. Moreover, by optional sampling,
	 \[\mathbb{E}_{\mathbf{x}}^{\mathbb{Q}^{{\bm \theta}^{c^\ast}}}\left[M_{\tau^\ast} \mathbbm{1}_{\{\tau^\ast<\infty\}}\right]=\mathbb{E}_{\mathbf{x}}^{\mathbb{Q}^{{\bm \theta}^{c^\ast}}}\left[M_0\right].\]
	 This yields that for $\tau=\tau^*$ both inequalities in the calculations above are indeed equations, i.e.
	 \begin{align*}
	 \mathbb{E}_{\mathbf{x}}^{\mathbb{Q}^{{\bm \theta}^{c^\ast}}}\left[e^{-r\tau^\ast}F(X_{\tau^\ast},Y_{\tau^\ast})\mathbbm{1}_{\{\tau^\ast<\infty\}}\right]
	 &=\Pi_{c^\ast}(z_i)yh_{c^\ast}(x/y).
	 \end{align*}
	 	 Since $\mathbb{Q}^{{\bm \theta}^{c^\ast}}\in \mathcal{P}^\kappa$ we notice that
	  $$
	 \inf_{\mathbb{Q}^{{\bm \theta}}\in \mathcal{P}^\kappa}\mathbb{E}_{\mathbf{x}}^{\mathbb{Q}^{{\bm \theta}}}\left[e^{-r\tau}F(X_\tau,Y_\tau)\right]
	 \leq \mathbb{E}_{\mathbf{x}}^{\mathbb{Q}^{{\bm \theta}^{c^\ast}}}\left[e^{-r\tau}F(X_\tau,Y_\tau)\right]
	 $$
	 implying that $V_\kappa(x,y)\leq \Pi_{c^\ast}(z_i^\ast)yh_{c^\ast}(x/y)$ for all $(x,y)\in \mathbb{R}_+^2$, i.e., the first inequality in (A). The calculations furthermore prove the second equilibrium condition in (C).

To prove the opposite inequality in (A) and the first equilibrium condition in (C) we obtain -- using again that the admissible stopping policy $\tau^\ast=\inf\{t\geq 0: Z_t\not\in (z_1^\ast,z_2^\ast)\}\in \mathcal{T}$ is $\mathbb{Q}^{{\bm \theta}^{c^\ast}}$-a.s. finite and $\Pi_c(Z_{\tau^\ast})\geq \left(\Pi_{c^\ast}(z_1^\ast)\wedge\Pi_{c^\ast}(z_2^\ast)\right)$ on the set $\tau^{\ast}<\infty$ as well as Lemma \ref{l1} --
\begin{align*}
V_\kappa(x,y) &\geq \inf_{\mathbb{Q}^{{\bm \theta}}\in \mathcal{P}^\kappa}\mathbb{E}_{\mathbf{x}}^{{\mathbb{Q}^{{\bm \theta}}}}\left[e^{-r\tau^{\ast}}F(X_{\tau^{\ast}},Y_{\tau^{\ast}})\mathbbm{1}_{\{\tau^{\ast}<\infty\}}\right]\\
&=\inf_{\mathbb{Q}^{{\bm \theta}}\in \mathcal{P}^\kappa}\mathbb{E}_{\mathbf{x}}^{{\mathbb{Q}^{{\bm \theta}}}}\left[e^{-r\tau^{\ast}}Y_{\tau^{\ast}}h_{c^\ast}(Z_{\tau^{\ast}})\Pi_{c^\ast}(Z_{\tau^{\ast}})\mathbbm{1}_{\{\tau^{\ast}<\infty\}}\right]\\
&\geq \left(\Pi_{c^\ast}(z_1^\ast)\wedge\Pi_{c^\ast}(z_2^\ast)\right)\inf_{\mathbb{Q}^{{\bm \theta}}\in \mathcal{P}^\kappa}\mathbb{E}_{\mathbf{x}}^{{\mathbb{Q}^{{\bm \theta}}}}\left[e^{-r\tau^{\ast}}Y_{\tau^{\ast}}h_{c^\ast}(Z_{\tau^{\ast}})\mathbbm{1}_{\{\tau^{\ast}<\infty\}}\right]\\
&= \left(\Pi_{c^\ast}(z_1^\ast)\wedge\Pi_{c^\ast}(z_2^\ast)\right)\mathbb{E}_{\mathbf{x}}^{{\mathbb{Q}^{{\bm \theta}^{c^\ast}}}}\left[e^{-r\tau^{\ast}}Y_{\tau^{\ast}}h_{c^\ast}(Z_{\tau^{\ast}})\mathbbm{1}_{\{\tau^{\ast}<\infty\}}\right]\\
&= \left(\Pi_{c^\ast}(z_1^\ast)\wedge\Pi_{c^\ast}(z_2^\ast)\right)yh_{c^\ast}(z)
\end{align*}
for all $(x,y)\in \{(x,y)\in \mathbb{R}_+^2:z_1^\ast<x/y < z_2^\ast\}$, proving (A). Now, (A) together with the last inequalities im Lemma \ref{l1}, yields the first part of (C).

Finally, noticing that for all $(x,y)\in \{(x,y)\in \mathbb{R}_+^2:z_1^\ast<x/y < z_2^\ast\}$ we have
$$
V_\kappa(x,y)-F(x,y) = yh_{c^\ast}(z)\left(\Pi_{c^\ast}(z^\ast)-\Pi_{c^\ast}(z)\right)>0
$$
showing that $\{(x,y)\in \mathbb{R}_+^2:z_1^\ast<x/y < z_2^\ast\}\subseteq C_\kappa$, viz. (B).
\end{proof}
According to Proposition \ref{prop:max_pt_stopping} the points belonging to the set $\argmax\{\Pi_c(z)\}$ are part of the stopping set where waiting is suboptimal independently of the reference point $c$. This is an interesting finding since it provides a straightforward technique for identifying elements in the stopping region.  Theorem \ref{t1} delineates a set of conditions under which the optimal stopping strategy constitutes a two boundary policy and the value can be expressed in terms of the function $h_c(z)$.

As in the case of a general reference point $c\in\mathbb{R}_+$ we again notice that the functions $h_0(z)$ and $h_\infty(z)$
can be utilized in the analysis of the optimal timing policy and the associated worst case measure for the one-sided boundary situation. This is summarized in the following two theorems.

\begin{theorem}[lower-boundary case]\label{t2}
	Assume that  condition $\mu_X-\mu_Y>\kappa(\sigma_X+\sigma_Y)+\frac{1}{2}(\sigma_X^2+\sigma_Y^2)$ is met and that there exists a point $z^\ast\in \argmax\left\{\Pi_{0}(z)\right\}\in(0,\infty)$. Then
	\begin{itemize}
		\item[(A)]
		$V_\kappa(x,y) =
		y\Pi_0(z^\ast)h_0(x/y)$ whenever $x<z^\ast y$.
		\item[(B)] If $\Pi_{0}(z^\ast)>\Pi_{0}(z)$ for all $z<z^\ast$, then $$\{(x,y)\in \mathbb{R}_+^2: x<z^\ast y\}\subseteq C_\kappa:=\{(x,y)\in \mathbb{R}_+^2:V_\kappa(x,y)>F(x,y)\}.$$
		\item[(C)] Let $\tau^\ast\in \mathcal T$ be such that  $\tau^\ast=\inf\{t\geq 0:Z_t>z^\ast\}$ $\mathbb P_{\bf x}$-a.s. for all initial points $(x,y)$ with $x<z^\ast y$. Then, $(\mathbb{Q}^{{\bm \theta}^{0}},\tau^\ast)$ is an equilibrium in the sense that for all initial points $(x,y)$ with $x < z^\ast y$ it holds that
		\begin{align*}
		\mathbb{E}_{\mathbf{x}}^{{\mathbb{Q}^{{\bm \theta}}}}\left[e^{-r\tau^{\ast}}F(X_{\tau^{\ast}},Y_{\tau^{\ast}})\mathbbm{1}_{\{\tau^{\ast}<\infty\}}\right]&\geq \mathbb{E}_{\mathbf{x}}^{\mathbb{Q}^{{\bm \theta}^{0}}}\left[e^{-r\tau^{\ast}}F(X_{\tau^{\ast}},Y_{\tau^{\ast}})\mathbbm{1}_{\{\tau^{\ast}<\infty\}}\right]\mbox{ for all }\mathbb{Q}^{{\bm \theta}}\in \mathcal{P}^\kappa\\
		\mathbb{E}_{\mathbf{x}}^{\mathbb{Q}^{{\bm \theta}^{0}}}\left[e^{-r\tau}F(X_{\tau},Y_{\tau})\mathbbm{1}_{\{\tau<\infty\}}\right]&\leq \mathbb{E}_{\mathbf{x}}^{\mathbb{Q}^{{\bm \theta}^{0}}}\left[e^{-r\tau^{\ast}}F(X_{\tau^{\ast}},Y_{\tau^{\ast}})\mathbbm{1}_{\{\tau^{\ast}<\infty\}}\right]\mbox{ for all }\tau\in\mathcal T.
		\end{align*}
	\end{itemize}
\end{theorem}

\begin{proof}
	Noticing that condition $\mu_X-\mu_Y>\kappa(\sigma_X+\sigma_Y)+\frac{1}{2}(\sigma_X^2+\sigma_Y^2)$ yields that $\tau^\ast=\inf\{t\geq 0:Z_t\geq z^\ast\}$ is $\mathbb{Q}^{{\bm \theta}^{0}}$-a.s. finite  is met for all initial points $(x,y)$ with $x<z^\ast y$,
	the statement holds by a straightforward modification of the arguments in the proof of Theorem \ref{t1}.
\end{proof}

Not surprisingly, we obtain analogously
\begin{theorem}[upper-boundary case]\label{t3}
	Assume that  condition $\mu_X-\mu_Y<\frac{1}{2}(\sigma_X^2+\sigma_Y^2)-\kappa(\sigma_X+\sigma_Y)$
	is met and that there exists a point $z^\ast\in \argmax\left\{\Pi_{\infty}(z)\right\}\in(0,\infty)$. Then
	\begin{itemize}
		\item[(A)]
		$V_\kappa(x,y) =
		y\Pi_\infty(z^\ast)h_\infty(x/y)$ whenever $x>z^\ast y$.
		\item[(B)] If $\Pi_{\infty}(z^\ast)>\Pi_{\infty}(z)$ for all $z>z^\ast$, then $$\{(x,y)\in \mathbb{R}_+^2: x>z^\ast y\}\subseteq C_\kappa:=\{(x,y)\in \mathbb{R}_+^2:V_\kappa(x,y)>F(x,y)\}.$$
		\item[(C)] Let $\tau^\ast\in \mathcal T$ be such that  $\tau^\ast=\inf\{t\geq 0:Z_t<z^\ast\}$ $\mathbb P_{\bf x}$-a.s. for all initial points $(x,y)$ with $x>z^\ast y$. Then, $(\mathbb{Q}^{{\bm \theta}^{\infty}},\tau^\ast)$ is an equilibrium in the sense that for all initial points $(x,y)$ with $x < z^\ast y$ it holds that
		\begin{align*}
		\mathbb{E}_{\mathbf{x}}^{{\mathbb{Q}^{{\bm \theta}}}}\left[e^{-r\tau^{\ast}}F(X_{\tau^{\ast}},Y_{\tau^{\ast}})\mathbbm{1}_{\{\tau^{\ast}<\infty\}}\right]&\geq \mathbb{E}_{\mathbf{x}}^{\mathbb{Q}^{{\bm \theta}^{\infty}}}\left[e^{-r\tau^{\ast}}F(X_{\tau^{\ast}},Y_{\tau^{\ast}})\mathbbm{1}_{\{\tau^{\ast}<\infty\}}\right]\mbox{ for all }\mathbb{Q}^{{\bm \theta}}\in \mathcal{P}^\kappa\\
		\mathbb{E}_{\mathbf{x}}^{\mathbb{Q}^{{\bm \theta}^{\infty}}}\left[e^{-r\tau}F(X_{\tau},Y_{\tau})\mathbbm{1}_{\{\tau<\infty\}}\right]&\leq \mathbb{E}_{\mathbf{x}}^{\mathbb{Q}^{{\bm \theta}^{\infty}}}\left[e^{-r\tau^{\ast}}F(X_{\tau^{\ast}},Y_{\tau^{\ast}})\mathbbm{1}_{\{\tau^{\ast}<\infty\}}\right]\mbox{ for all }\tau\in\mathcal T.
		\end{align*}
	\end{itemize}
\end{theorem}

\begin{remark}
	Proposition \ref{prop:max_pt_stopping} together with the previous three theorems leads to a procedure for solving general versions of our stopping problem \eqref{stopping} as follows:
	\begin{enumerate}
		\item Determine the maximum points of the real functions $\Pi_{c}$ for all $c\in[0,\infty]$ to find
		$$\tilde\Gamma:=\bigcup_{c\in[0,\infty]}\big\{(x,y)\in\mathbb R_+^2:x/y\in \argmax\{\Pi_{c}(z)\}\big\}\subseteq \Gamma_\kappa.$$
		\item The complement of $\tilde\Gamma$ forms a partition of $\mathbb R_+^2$ into subcones. For each such cone $\{(x,y)\in \mathbb{R}_+^2: z_1^\ast y < x < z_2^\ast y\}$,  $z_1,z_2\in[0,\infty]$, find a parameter $c^*\in[0,\infty]$ such that Theorem \ref{t1}, \ref{t2}, or \ref{t3} is applicable. Then, $\mathbb{Q}^{{\bm \theta}^{c^\ast}}$ is a worst case prior for the corresponding subcone, $\Pi_{c^\ast}(z_i^\ast)yh_{c^\ast}(x/y)$ defines the value function on it, and the first entrance time into $\tilde\Gamma$ is a (global) optimal stopping time for \eqref{stopping}. For last claim, note that the underlying processes do not have jumps and therefore, the different connected components of the complement of $\tilde\Gamma$ do not communicate with each other.
	\end{enumerate}
	
The question arises whether this procedure can always be applied. In other words: Is one of the Theorem \ref{t1}, \ref{t2}, or \ref{t3} always applicable in Step 2?
	 -- The answer is \emph{yes} under some mild assumptions. To see that this is indeed the case, assume that condition $r>\max(\mu_X-\kappa\sigma_X,\mu_Y-\kappa\sigma_Y)$ holds and consider, for fixed $(x,y)\in\mathbb R_+^2$, $z:=x/y$, the functions $\sup_{\y\geq z}\Pi_c(\y)$ and $\sup_{\y\leq z}\Pi_c(\y)$ as functions of $c$.
It is clear that
\begin{align*}
\sup_{\y\geq z}\Pi_c(\y) &= \left(\inf_{\y\geq z}\left[\frac{H_{1c}(\y)}{F(\y,1)}\right] \vee \inf_{\y\geq z}\left[\frac{H_{2c}(\y)}{F(\y,1)}\right]\vee \inf_{\y\geq z}\left[\frac{H_{3c}(\y)}{F(\y,1)}\right]\right)^{-1}\\
\sup_{\y\leq z}\Pi_c(\y) &= \left(\inf_{\y\leq z}\left[\frac{H_{1c}(\y)}{F(\y,1)}\right] \vee \inf_{\y\leq z}\left[\frac{H_{2c}(\y)}{F(\y,1)}\right]\vee \inf_{\y\leq z}\left[\frac{H_{3c}(\y)}{F(\y,1)}\right]\right)^{-1}
\end{align*}
Since, for example,
$$
\inf_{\y\leq z}\left[\frac{H_{3c}(\y)}{F(\y,1)}\right]=c^{-\psi_\kappa}\inf_{\y\leq z}\left[\frac{1}{F(\y,1)}\frac{\hat{\psi}_\kappa}{\hat{\psi}_\kappa-1}l^{\hat{\varphi}_\kappa}\left(\frac{1-\varphi_\kappa}{\psi_\kappa-\varphi_\kappa}\left(\frac{\y}{l}\right)^{\psi_\kappa} + \frac{\psi_\kappa-1}{\psi_\kappa-\varphi_\kappa}\left(\frac{\y}{l}\right)^{\varphi_\kappa}c^{\psi_\kappa-\varphi_\kappa}\right)\right]
$$
we observe now by utilizing the fact that the pointwise infimum of an affine function is concave and, thus, continuous, that $\inf_{\y\leq z}\left[\frac{H_{3c}(\y)}{F(\y,1)}\right]$ is continuous as a function of $c$. Since the maximum of continuous functions is continuous, we notice that $\sup_{\y\leq z}\Pi_c(\y)$ is continuous as a function of $c$ too. The same argument is naturally valid for $\sup_{\y\geq z}\Pi_c(\y)$ as well. Consequently, the difference
$$
D(c)=\sup_{\y\geq z}\Pi_c(\y)-\sup_{\y\leq z}\Pi_c(\y)
$$
is continuous as a function of $c\in [0,\infty]$. Now, there are three possible cases:
\begin{enumerate}
	\item $D(0)\geq 0$. Then there exists $z^\ast\in \argmax\left\{\Pi_{0}(\y)\right\}\in[z,\infty)$ such that Theorem \ref{t2} is applicable whenever $\sup_{\y\geq z}\Pi_0(\y)$ is a maximum (which  is the case under standard continuity- and growths-assumptions on the reward function $F$).
	\item $D(\infty)\leq 0$. Then, whenever $\sup_{\y\leq z}\Pi_\infty(\y)$ is a maximum, there exists $z^\ast\in \argmax\left\{\Pi_{\infty}(\y)\right\}\in(0,z]$ such that Theorem \ref{t3} is applicable.
	\item $D(0)<0<D(\infty).$ By continuity, there exists $c^\ast\in \mathbb{R}_+$ such that $D(c^\ast)=0$. Assuming again that the suprema are maxima, there exist two points $z_1<z<z_2$ such that Theorem \ref{t1} is applicable.
\end{enumerate}
%

It is at this point worth emphasizing that our findings indicate that the representation of the value as the smallest element of an appropriately chosen function space developed in \cite{Chr13} applies in some cases in the present setting as well.
In that case we have
$$
V_\kappa(x,y)=y\inf\{\lambda h_c(x/y):c\in [0,\infty],\lambda\in[0,\infty],\lambda h_c(x/y)\geq F(x/y,1)\},
$$
for all $x/y\in \mathbb{R}_+$.
\end{remark}

It is clear from the description above that the minimal harmonic functions in the present case differ from the ones appearing in the one-dimensional setting. There are two main reasons for this. First, the presence of two driving random factor dynamics implies that the underlying density generators may separately switch from one extreme to another at different states as is clear from the form of the sets $A_1,A_2,$ and $A_3$. Second, the assumed homogeneity of the exercise payoff implies that both the growth rate as well as the density generator associated with the numeraire variable $Y$ affect the rate at which the problem is discounted. This mechanism where nature also selects the rate at which the problem is discounted cannot naturally appear in a one-dimensional setting where discounting is not affected by the characteristics of the underlying factor dynamics.

\section{Explicit Illustrations}

\subsection{Compound Option}

In order to illustrate an option where the stopping region constitutes a closed interval on $\mathbb{R}_+$ we now consider the compound option case
$$F(x,y)=\min\left((x-Ky)^+,(My-x)^+\right) = y \min\left((z-K)^+,(M-z)^+\right),$$ where the strike prices satisfy the inequality $M>K>0$.
We also assume that $r>\max(\mu_X-\kappa\sigma_X,\mu_Y-\kappa\sigma_Y)$ guaranteeing that $\psi_\kappa>1$ and $\varphi_{-\kappa}<0$.
Since $F(x,y)=y(z-K)^+$ on $z<L:=\frac{1}{2}(K+M)$ we notice by solving the optimization problem
$$
\max_{z\in[0,L]}\frac{(z-K)^+}{z^{\psi_\kappa}}
$$
that the lower optimal exercise boundary is
$$
z_1^\ast=\frac{\psi_\kappa}{\psi_\kappa-1}K\wedge L.
$$
Analogously, since $F(x,y)=y(M-z)^+$ on $z>L$ we notice by solving the optimization problem
$$
\max_{z\in[L,\infty]}\frac{(M-z)^+}{z^{\varphi_{-\kappa}}}
$$
that the upper optimal exercise boundary is
$$
z_2^\ast=\frac{\varphi_{-\kappa}}{\varphi_{-\kappa}-1}M \vee L.
$$
{Utilizing Proposition \ref{prop:max_pt_stopping} shows that $[z_1^*,z_2^*]$ is a subset of the optimal stopping set. Therefore, }the value reads as $V_\kappa(x,y) = y v(z)$, where
$$
v(z)=\begin{cases}
(M-z_2^\ast)\left(\frac{z}{z_2^\ast}\right)^{\varphi_{-\kappa}}&z> z_2^\ast\\
(z-K)^+\wedge(M-z)^+&z_1^\ast\leq z \leq z_2^\ast\\
(z_1^\ast-K)\left(\frac{z}{z_1^\ast}\right)^{\psi_\kappa}&z<z_1^\ast.
\end{cases}
$$
Moreover, the optimal density generators now read as
$$
(\theta_1^\ast,\theta_2^\ast)=\begin{cases}
(-\kappa,\kappa),&z> L\\
(\kappa,-\kappa),&z< L.
\end{cases}
$$
It is at this point worth emphasizing that the impact of ambiguity on the optimal timing policy is asymmetric despite its seemingly simple symmetric structure. To see that this is indeed the case, we first observe that under our assumptions $\partial \psi_\kappa/\partial \kappa > 0$ and $\partial \varphi_{-\kappa}/\partial \kappa < 0$. Given these sensitivities, we now investigate when a corner solution may arise. To this end we first observe that $z_1^\ast=L$ whenever $\psi_\kappa \leq (M+K)/(M-K)$ and $z_2^\ast=L$ whenever $\varphi_{-\kappa} \geq (M+K)/(K-M)$. Since $\partial z_1^\ast/\partial \kappa \leq 0$, $\partial z_2^\ast/\partial \kappa \geq 0$, $\lim_{\kappa\rightarrow\infty}\psi_\kappa=+\infty$, and $\lim_{\kappa\rightarrow\infty}\varphi_{-\kappa}=-\infty$, we notice that if a corner solution is attained in the unambiguous benchmark setting, that is, if $\psi_0 < (M+K)/(M-K)$ and $\varphi_{0} > (M+K)/(K-M)$, then there exists two critical degrees $\kappa_1,\kappa_2$ so that $z_1^\ast = L$ for $\kappa\leq \kappa_1$ and $z_2^\ast=L$ for $\kappa\leq \kappa_2$. We also notice that the optimal timing rules approach the {\em Marshallian investment rule} in the limit as $\kappa\rightarrow\infty$. More precisely, $\lim_{\kappa\rightarrow\infty}z^{\ast}_{1}=K$ and $\lim_{\kappa\rightarrow\infty}z^{\ast}_{2}=M$. The boundaries are illustrated as functions of the degree of ambiguity
in Figure \ref{compound} under the parameter assumptions $r = 0.0351, \sigma_X = \sigma_Y = 0.1, \mu_X = \mu_Y = 0.035, K=1,$ and $M=2$ (implying that $\kappa_1\approx 0.1198$ and $\kappa_2\approx 0.171286$).
\begin{figure}[!ht]
\begin{center}
\includegraphics[width=0.6\textwidth]{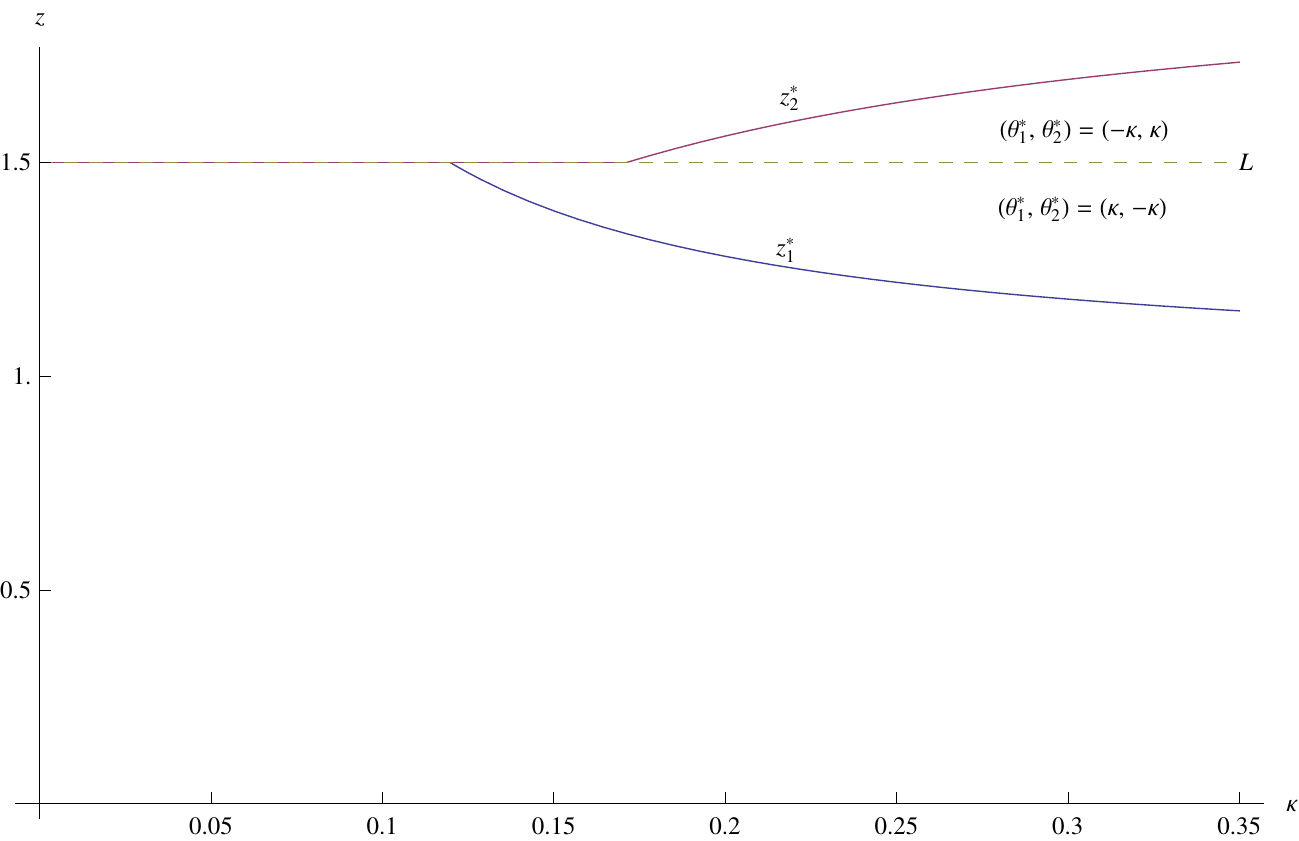}
\caption{\small The optimal exercise ratios $z_1^{\ast}, z_2^{\ast}$ as functions of $\kappa$}\label{compound}
\end{center}
\end{figure}

The impact of increased volatility on the optimal exercise boundaries depends on the exact parametrization of the problem and, therefore, on the degree of ambiguity. To see that this is indeed the case, we first notice that
\begin{align*}
\frac{\partial \psi_\kappa}{\partial \sigma_Y} &=\frac{2(\kappa-\sigma_Y\psi_\kappa)(\psi_\kappa-1)}{(\sigma_X^2+\sigma_Y^2)(\psi_\kappa-\varphi_\kappa)} \gtreqqless 0,\quad \kappa\gtreqqless\sigma_Y\psi_\kappa,\\
\frac{\partial \psi_\kappa}{\partial \sigma_X} &=\frac{2\psi_\kappa(\kappa-\sigma_X(\psi_\kappa-1))}{(\sigma_X^2+\sigma_Y^2)(\psi_\kappa-\varphi_\kappa)}\gtreqqless 0,\quad
\kappa\lesseqqgtr\sigma_X(\psi_\kappa-1),\\
\frac{\partial \varphi_{-\kappa}}{\partial \sigma_Y} &=\frac{2(\varphi_{-\kappa}-1)(\sigma_Y\varphi_{-\kappa}+\kappa)}{(\sigma_X^2+\sigma_Y^2)(\psi_{-\kappa}-\varphi_{-\kappa})} \gtreqqless 0,\quad \sigma_Y\varphi_{-\kappa}\lesseqqgtr-\kappa,\\
\frac{\partial \varphi_{-\kappa}}{\partial \sigma_X} &=\frac{2\varphi_{-\kappa}(\kappa-\sigma_X(1-\varphi_{-\kappa}))}{(\sigma_X^2+\sigma_Y^2)(\psi_{-\kappa}-\varphi_{-\kappa})}\gtreqqless 0,\quad
\kappa\lesseqqgtr\sigma_X(1-\varphi_{-\kappa}).
\end{align*}
Since
\begin{align*}
\frac{\partial z_1^\ast}{\partial\psi_\kappa} &=\frac{-K}{(\psi_\kappa-1)^2}<0\\
\frac{\partial z_2^\ast}{\partial\varphi_{-\kappa}} &=\frac{-M}{(\varphi_{-\kappa}-1)^2}<0
\end{align*}
whenever $z_1^\ast < L < z_2^\ast$, we notice that the impact of higher volatility on the optimal exercise thresholds is ambiguous.
This is explicitly illustrated in Figure \ref{compoundvolatility} under the parameter assumptions $r = 0.0351, \sigma_X = 0.05, \mu_X = \mu_Y = 0.035, K=1,$ and $M=2$. As is clear from this figure, the impact of increased volatility on optimal timing is not necessarily decelerating in the presence of ambiguity.
\begin{figure}[!ht]
\begin{center}
\includegraphics[width=0.6\textwidth]{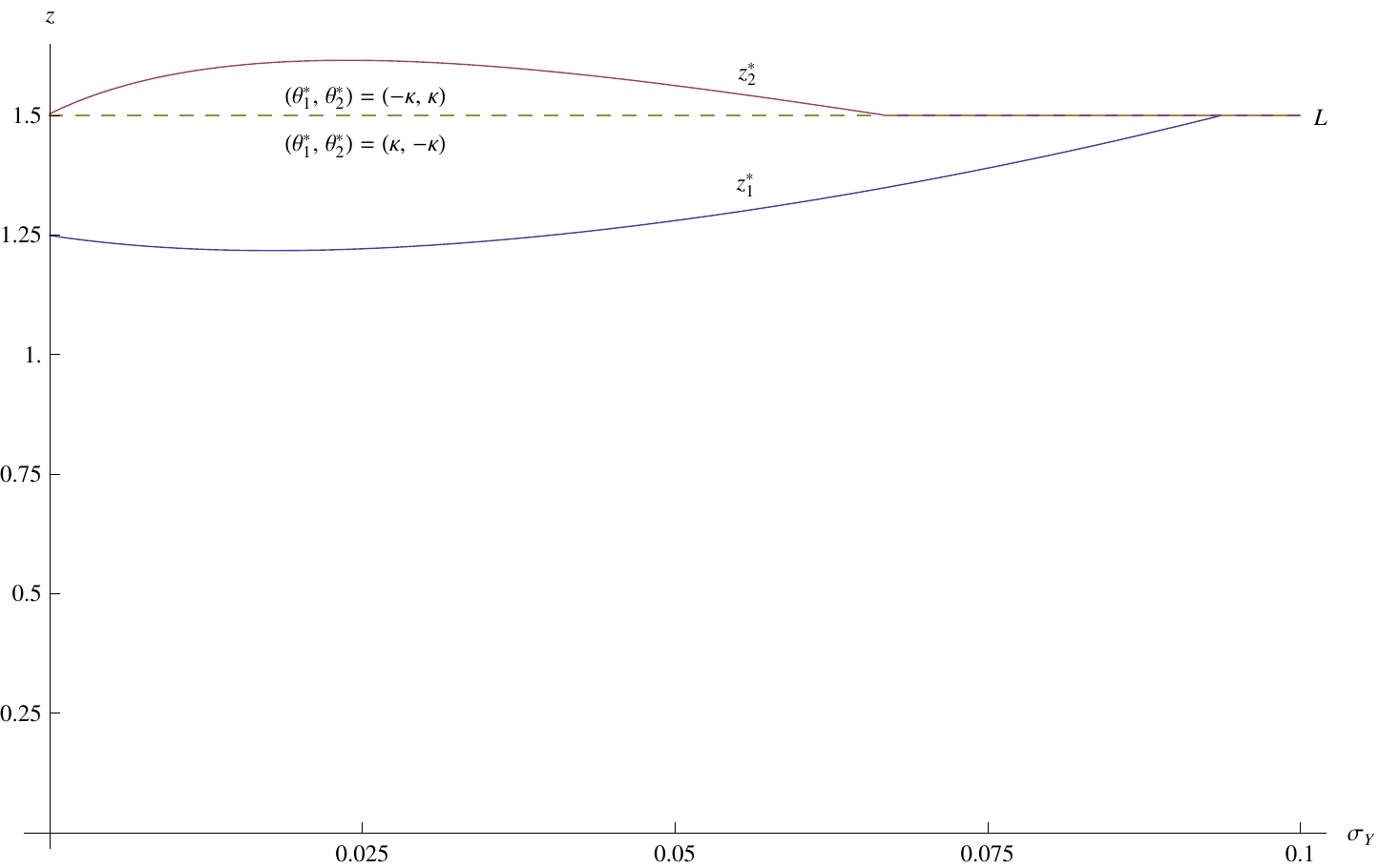}
\caption{\small The optimal exercise ratios $z_1^{\ast}, z_2^{\ast}$ as functions of $\kappa$}\label{compoundvolatility}
\end{center}
\end{figure}

\subsection{Floor Option}

In order to illustrate a case where the determination of the worst case measure involves the determination of the switching point at which the drift of the underlying is taken from one extreme to another let us consider the reward $F(x,y)=\max(x,y)$ (cf. \cite{GuSh01}). Assume that  $r>\max(\mu_X-\kappa\sigma_X,\mu_Y-\kappa\sigma_Y)$. In line with our general results we now analyze the ratio
$$
\Pi_c(z):=\frac{\max(z,1)}{h_{c}(z)},
$$
where $h_{c}(z)$ is defined in \eqref{harmonic}. It is now a straightforward exercise in ordinary differentiation to show that
a local maximum is attained on the set $z\geq c\vee 1$ at
$$
\frac{z^{\ast}_{2}}{c}=l=\left(\frac{\hat{\psi}_\kappa(1-\hat{\varphi}_\kappa)}
{\hat{\varphi}_\kappa(1-\hat{\psi}_\kappa)}\right)^{\frac{1}{\hat{\psi}_\kappa-\hat{\varphi}_\kappa}}>1.
$$
Plugging this into the ratio yields
$$
\Pi_c(z^{\ast}_{2})=\frac{\hat{\psi}_\kappa-1}{\hat{\psi}_\kappa}\left(\frac{\hat{\psi}_\kappa (1-\hat{\varphi}_\kappa)}{\hat{\varphi}_\kappa(1-\hat{\psi}_\kappa)}\right)^{\frac{1-\hat{\varphi}_\kappa}{\hat{\psi}_\kappa-\hat{\varphi}_\kappa}} c.
$$
On the other hand, since $h_{c}(z)$ is monotonically decreasing on $(0,c)$ the ratio $\Pi_c(z)$ is increasing on $(0,c)$. Consequently, the
lower optimal boundary is attained at $z_1^\ast=c^\ast$. Therefore, we have
$$
1=\frac{\hat{\psi}_\kappa-1}{\hat{\psi}_\kappa}\left(\frac{\hat{\psi}_\kappa (1-\hat{\varphi}_\kappa)}{\hat{\varphi}_\kappa(1-\hat{\psi}_\kappa)}\right)^{\frac{1-\hat{\varphi}_\kappa}{\hat{\psi}_\kappa-\hat{\varphi}_\kappa}} c^\ast
$$
yielding
$$
z_1^\ast = c^\ast = \frac{\hat{\psi}_\kappa}{\hat{\psi}_\kappa-1}\left(\frac{\hat{\psi}_\kappa (1-\hat{\varphi}_\kappa)}{\hat{\varphi}_\kappa(1-\hat{\psi}_\kappa)}\right)^{-\frac{1-\hat{\varphi}_\kappa}{\hat{\psi}_\kappa-\hat{\varphi}_\kappa}}
$$
and
$$
z_2^\ast = \frac{\hat{\psi}_\kappa}{\hat{\psi}_\kappa-1}\left(\frac{\hat{\psi}_\kappa (1-\hat{\varphi}_\kappa)}{\hat{\varphi}_\kappa(1-\hat{\psi}_\kappa)}\right)^{\frac{\hat{\varphi}_\kappa}{\hat{\psi}_\kappa-\hat{\varphi}_\kappa}}
$$
{In this case, the value of the optimal policy reads}
\begin{align*}
V_\kappa(x,y)=\begin{cases}
x, &x\geq z_2^\ast y\\
yh_{c^{\ast}}(x/y),&x/y\in (z_1^\ast,z_2^\ast)\\
y,&x\leq z_1^\ast y.
\end{cases}
\end{align*}
{From this expression we notice that the stopping region now reads as $\Gamma_\kappa=\{(x,y)\in \mathbb{R}_+^2:x\leq z_1^\ast y \textrm{ or } x\geq z_2^\ast y\}$.}
Moreover, we clearly have $\lim_{\kappa\downarrow 0}V_\kappa(x,y)=V_0(x,y)$, $\lim_{\kappa\downarrow 0}z_1^\ast=z_1$, and $\lim_{\kappa\downarrow 0}z_2^\ast=z_2$.

It is at this point worth emphasizing that the worst case measure is characterized by the density generators
$$
(\theta_{1t}^\ast,\theta_{2t}^\ast)=\begin{cases}
(\kappa,-\kappa),&X_t\geq z_2^\ast Y_t\\
(\kappa,\kappa), & z_1^\ast Y_t\leq X_t < z_2^\ast Y_t\\
(-\kappa,\kappa),&X_t<z_1^\ast Y_t.
\end{cases}
$$
Interestingly, the worst case measure is such that it pushes on $x<z_1^\ast y$ the underlying process $X$ upwards towards the set $x\geq z_1^\ast y$ where the drift once again switches pushing the process back towards the set $x<z_1^\ast y$. Analogously, the worst case measure pushes on $x>z_2^\ast y$ the underlying process $Y$ towards the set $x<z_2^\ast y$ where the drift once again switches and pushes $Y_t$ back towards the set $x<z_2^\ast y$.

The optimal exercise boundaries $z_1^\ast,z_2^\ast$ are illustrated in Figure \ref{figguoshep} as functions of the degree of ambiguity $\kappa$ under the assumptions that
$\mu_X = 2\%, \mu_Y = 4\%, \sigma_X = 5\%, \sigma_Y = 10\%,$ and $r = 5\%$.
\begin{figure}[h!]
\begin{center}
\includegraphics[width=0.5\textwidth]{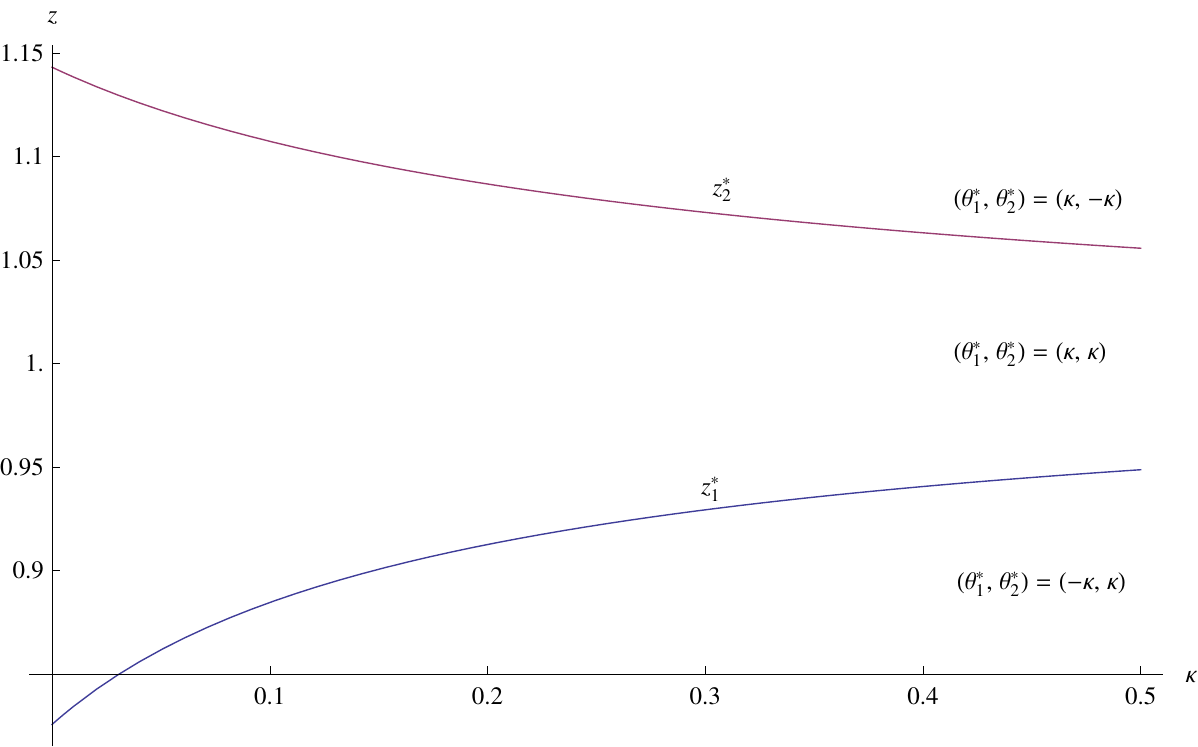}
\end{center}
\caption{\small{The optimal exercise boundaries $z_1^\ast,z_2^\ast$}}\label{figguoshep}
\end{figure}

\subsection{Straddle Option}
In order to illustrate a nontrivial two-boundary case where the threshold where the density generators switch from one extreme to another does not coincide with one of the optimal boundaries, we now consider the straddle option $F(x,y)=|x-y|$ under the assumption $r>\max(\mu_X-\kappa\sigma_X,\mu_Y-\kappa\sigma_Y)$. It is clear that in this case the candidate boundaries have to be determined from the first order optimality conditions
\begin{align}
h_c(z_i^\ast) - h_c'(z_i^\ast) (z_i^\ast-1) &=0, \quad i=1,2.\label{elasticity1}
\end{align}
Utilizing now the strict convexity, smoothness, and limiting behavior of the function $h_{c}(z)$ shows that for all $c\in \mathbb{R}_+$ optimality condition \eqref{elasticity1} has a unique root $z_2^\ast\in (1\vee lc,\infty)$ and that $z_2^\ast$ is increasing as a function of $c$. Analogously, we also notice that for all $c\in \mathbb{R}_+$ optimality condition \eqref{elasticity1} has a unique root $z_1^\ast\in (0,1\wedge c)$ and that $z_1^\ast$ is increasing as a function of $c$. It is also clear that
$$
\lim_{c\downarrow 0}z_2^\ast = \frac{\psi_\kappa}{\psi_\kappa-1},\quad \lim_{c\rightarrow \infty}z_2^\ast = \infty,\quad \lim_{c\downarrow 0}z_1^\ast = 0,\quad \lim_{c\rightarrow \infty}z_1^\ast = \frac{\varphi_{-\kappa}}{\varphi_{-\kappa}-1}.
$$
Combining these observations with the limiting behavior of $h_{c}(z)$ shows that $\lim_{c\downarrow 0}\Pi_c(z_2^\ast)=0$, $\lim_{c\downarrow 0}\Pi_c(z_1^\ast)=1$, $\lim_{c\rightarrow\infty}\Pi_c(z_2^\ast)=\infty$, and $\lim_{c\rightarrow\infty}\Pi_c(z_1^\ast)=0$. Consequently, there is at least one $c^\ast\in \mathbb{R}_+$ so that $\Pi_{c^\ast}(z_1^\ast)=\Pi_{c^\ast}(z_2^\ast)$. Combining this observation with the monotonicity of $\Pi_c(z_1^\ast)$ and $\Pi_c(z_2^\ast)$ as a function of $c$ then
prove that $c^\ast$ is unique and that $z_1^\ast, z_2^\ast$ constitute the optimal stopping boundaries satisfying the inequality $z_2^\ast > lc^\ast > c^\ast > z_1^\ast$. 
{We notice that in the present example the value of the optimal policy reads}
\begin{align*}
V_\kappa(x,y)=\begin{cases}
x-y, &x\geq z_2^\ast y\\
y(z_2^\ast-1)h_{c^{\ast}}(x/y),&x/y\in (z_1^\ast,z_2^\ast)\\
y-x,&x\leq z_1^\ast y.
\end{cases}
\end{align*}
{From this expression we notice that the stopping region again reads as $\Gamma_\kappa=\{(x,y)\in \mathbb{R}_+^2:x\leq z_1^\ast y \textrm{ or } x\geq z_2^\ast y\}$.}

These quantities are illustrated as functions of the degree of ambiguity $\kappa$ in Figure \ref{straddleexample} under
the assumptions that $\mu_X = 2.5\%, \mu_Y = 3\%, \sigma_X = 7.5\%, \sigma_Y = 10\%$, and $r = 3.5\%$.
\begin{figure}[h!]
\begin{center}
\includegraphics[width=0.5\textwidth]{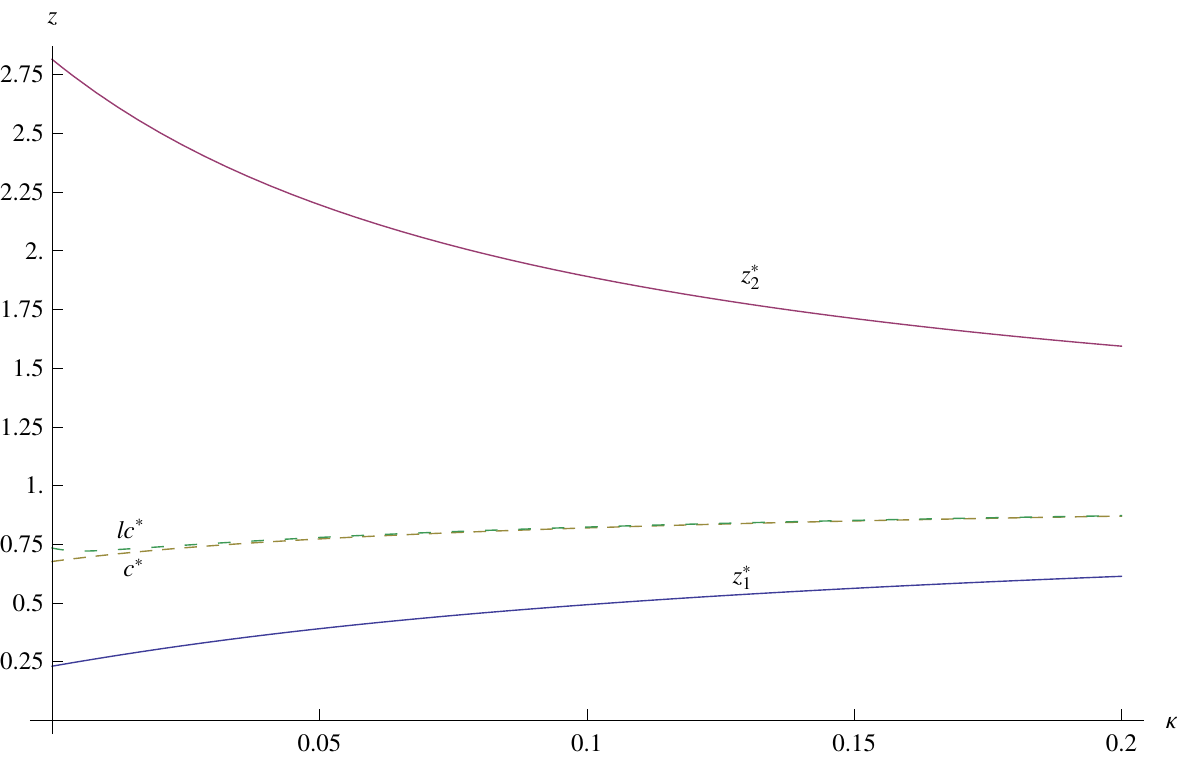}
\end{center}
\caption{\small{The optimal exercise boundaries $z_1^\ast,z_2^\ast$ and switching states $c^\ast$ and $lc^\ast$}}\label{straddleexample}
\end{figure}
It is at this point worth noticing that in the present example the optimal density generators read as
$$
(\theta_{1t}^\ast,\theta_{2t}^\ast)=\begin{cases}
(\kappa,-\kappa),&X_t\geq l c^\ast Y_t\\
(\kappa,\kappa), & c^\ast Y_t\leq X_t < l c^\ast Y_t\\
(-\kappa,\kappa),&X_t<c^\ast Y_t.
\end{cases}
$$
Thus, in contrast with the floor option case, the states at which the density generators switch form one extreme state to another do not coincide with the exercise thresholds. An illustration can be found in Figure \ref{fig:straddle_drift}.
\begin{figure}[h!]
\begin{center}
\includegraphics[width=0.5\textwidth]{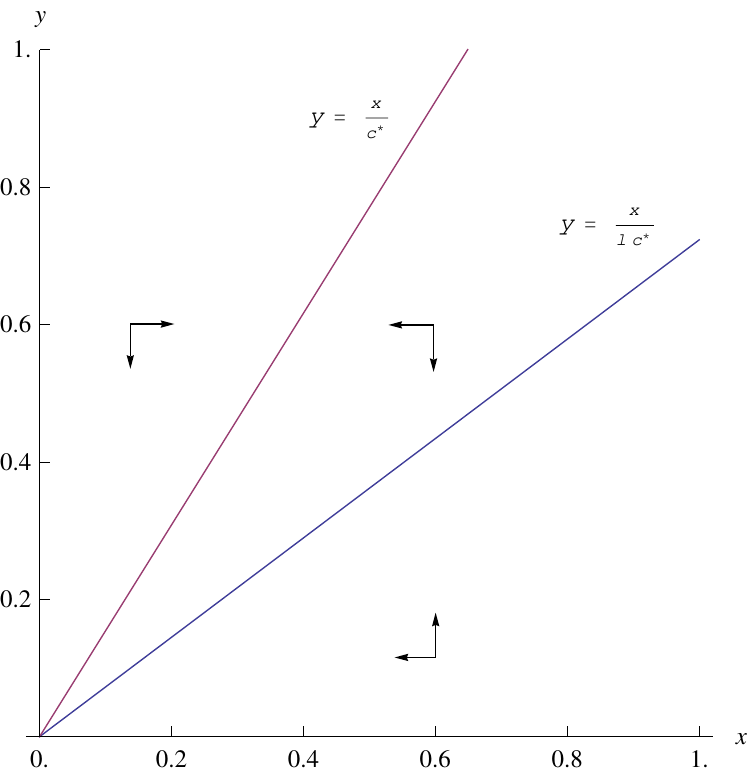}
\end{center}
\caption{\small{Illustration of the worst case drift}}\label{fig:straddle_drift}
\end{figure}



\subsection{Digital Option}

In order to illustrate our findings in a discontinuous framework, we now focus on a digital option with an exercise payoff
$$
F(x,y) = \begin{cases}
x,&x\geq ky\\
y,&x<ky,
\end{cases}
$$
where $k\in (0,1)$ is a known constant. Note that we rule out the case $k=1$ since that coincides with the floor option case treated earlier. Since
$F(x,y)=y(z \mathbbm{1}_{[k,\infty)}(z)+\mathbbm{1}_{(0,k)}(z))$, we notice that in this case
we have to focus on the ratio
$$
\Pi_c(z)=\frac{z \mathbbm{1}_{[k,\infty)}(z)+\mathbbm{1}_{(0,k)}(z)}{h_c(z)}.
$$
Standard analysis yields that if condition $r>\max(\mu_X-\kappa\sigma_X,\mu_Y-\kappa\sigma_Y)$ holds, then
$$
\Upsilon_{\ast}(c):=\sup_{y\leq k}\left\{\frac{1}{h_c(y)}\right\} = \begin{cases}
\frac{1}{h_c(k)},&c\geq k,\\
1,&c\leq k,
\end{cases}
$$
and
$$
\Upsilon^\ast(c):=\sup_{y\geq k}\left\{\frac{y}{h_c(y)}\right\} = \begin{cases}
\frac{lc}{h_c(lc)},&lc\geq k,\\
\frac{k}{h_c(k)},&lc\leq k.
\end{cases}
$$
Noticing now that $\Upsilon_{\ast}(c)=1$ for all $c\leq k$, $\lim_{c\rightarrow\infty}\Upsilon_{\ast}(c)=0,\lim_{c\rightarrow\infty}\Upsilon^\ast(c)=\infty$, and
$$
\Upsilon^\ast(k/l)=\frac{k}{h_{k/l}(k)}<1
$$
demonstrates that equation $\Upsilon^\ast(c)=\Upsilon_{\ast}(c)$ has at least one root $c^\ast\in (k/l,\infty)$. The monotonicity of $h_c(z)$ as a function of $c$ implies that $c^\ast$ is unique.
Moreover, we notice that $z_2^\ast=lc^\ast$ and
$$
z_1^\ast=\begin{cases}
k,&c^\ast\geq k,\\
c^\ast,&c^\ast\leq k.
\end{cases}
$$
In light of our results on the floor option, we notice immediately that $z_1^\ast=\hat{c}^\ast_\kappa, z_2^\ast=l\hat{c}^\ast_\kappa,$ and
$$
c^\ast=\hat{c}^\ast_\kappa=\frac{\hat{\psi}_\kappa}{\hat{\psi}_\kappa-1}\left(\frac{\hat{\psi}_\kappa (1-\hat{\varphi}_\kappa)}{\hat{\varphi}_\kappa(1-\hat{\psi}_\kappa)}\right)^{-\frac{1-\hat{\varphi}_\kappa}{\hat{\psi}_\kappa-\hat{\varphi}_\kappa}}
$$
as long as inequality $\hat{c}^\ast_\kappa\leq k$ holds. If, however, $\hat{c}^\ast_\kappa> k$, then $z_1^\ast=k,z_2^\ast = l c^\ast$ and the optimal $c^\ast$ constitutes the unique root
of equation
$$
\frac{1}{h_{c^\ast}(k)}=\frac{\hat{\psi}_\kappa-1}{\hat{\psi}_\kappa}\left(\frac{\hat{\psi}_\kappa (1-\hat{\varphi}_\kappa)}{\hat{\varphi}_\kappa(1-\hat{\psi}_\kappa)}\right)^{\frac{1-\hat{\varphi}_\kappa}{\hat{\psi}_\kappa-\hat{\varphi}_\kappa}}c^\ast
$$
on the set $(k,\infty)$. {Interestingly, despite the discontinuity of exercise payoff, we find that the value reads
in this case as}
$$
V_\kappa(x,y)=\begin{cases}
x,&x\geq z_2^\ast y,\\
\Pi_{c^\ast}(z_2^\ast)h_{c^\ast}(x/y),&z_1^\ast y<x<z_2^\ast y,\\
y,&x\leq z_1^\ast y.
\end{cases}
$$
{The stopping region has a similar structure with the one arising in the two previous examples.}

The value is illustrated in Figure \ref{discoexample} under the assumptions $\mu_X = 0.02, \mu_Y = 0.04, \sigma_X = 0.05, \sigma_Y = 0.1, r = 0.041, \kappa=0.28$, and $k=0.85$ (implying that $z_1^\ast=0.85, c^\ast = 0.899722,$ and $z_2^\ast=1.0877$).
\begin{figure}[h!]
\begin{center}
\includegraphics[width=0.5\textwidth]{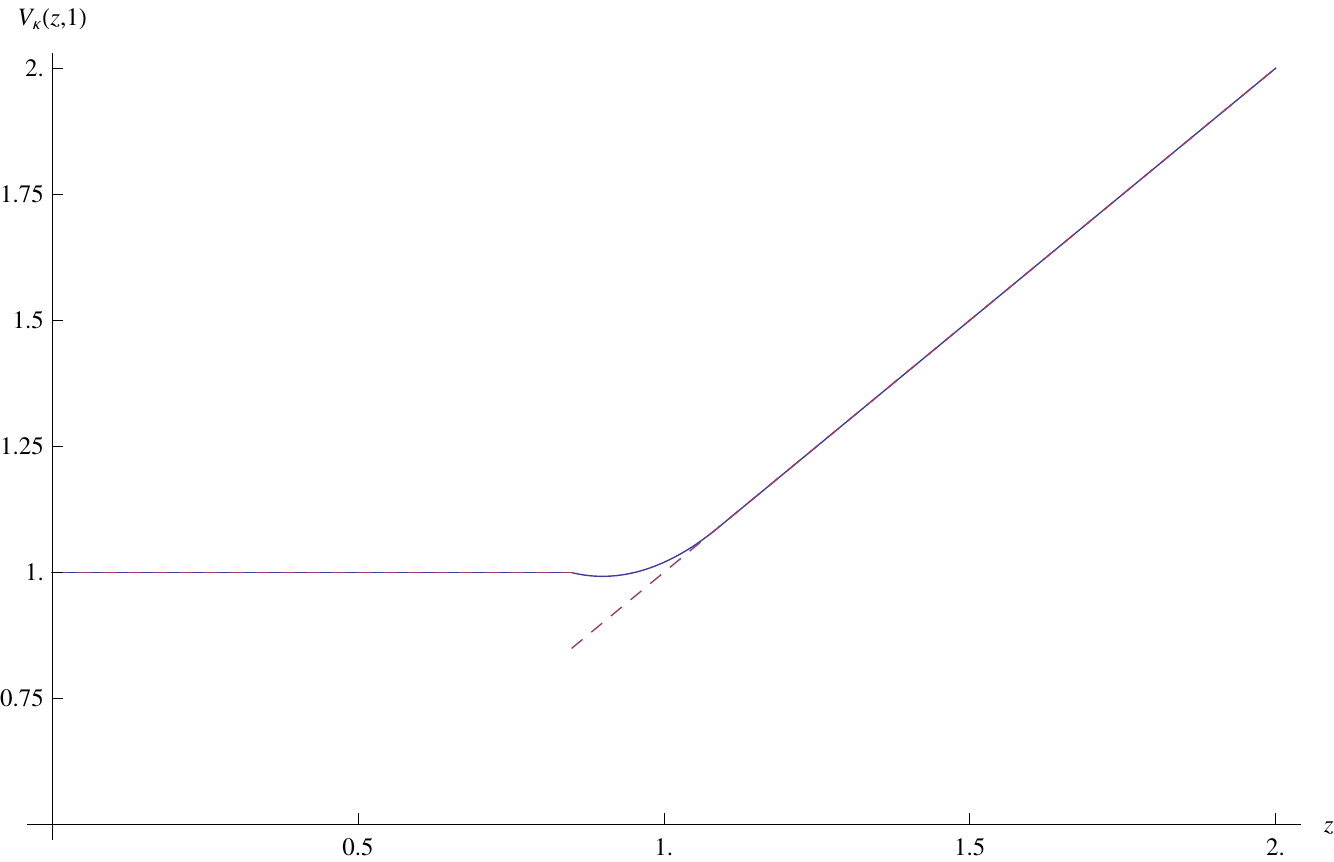}
\end{center}
\caption{\small{The value (uniform) of the optimal policy with discontinuous payoff (dashed)}}\label{discoexample}
\end{figure}
The optimal boundaries are, in turn, illustrated in Figure \ref{discobounds} as functions of the degree of ambiguity under the assumptions $\mu_X = 0.02, \mu_Y = 0.04, \sigma_X = 0.05, \sigma_Y = 0.1, r = 0.041$, and $k=0.85$.
\begin{figure}[h!]
\begin{center}
\includegraphics[width=0.5\textwidth]{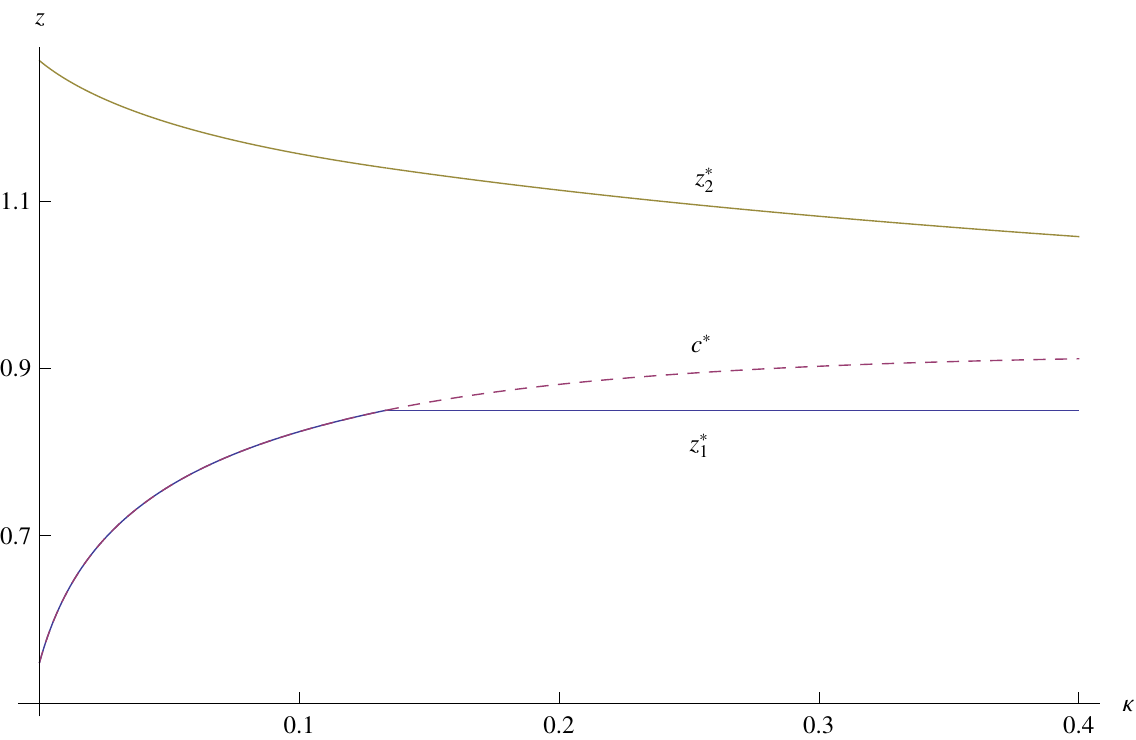}
\end{center}
\caption{\small{The optimal stopping boundaries in the discontinuous case}}\label{discobounds}
\end{figure}
\section{Conclusions}

We analyzed the impact of ambiguity on optimal stopping in the case where the exercise payoff is positively homogeneous and the underlying diffusions are two geometric Brownian motions. Utilizing the fact that the ratio of two geometric Brownian motions constitutes a geometric Brownian motion we reduced the dimensionality of the problem and extended the approach based on minimal excessive functions developed in \cite{Chr13} to the considered setting. Since a positively homogeneous function is not necessarily continuous, our results cast light on the optimal policy and its value also in nonsmooth cases.

There are several interesting directions towards which our analysis could naturally be extended. First, even though geometric Brownian motion constitutes the key benchmark process in financial and economic applications of optimal stopping, it would be of interest to study whether our principal conclusions would remain valid in a more general setting. The same argument is valid for the chosen payoff structure as well. Even though linearly homogeneous payoff structures play a prominent role in economic applications, analyzing the impact of more complex payoffs would be an interesting direction towards which our analysis could be extended. Both proposed extensions are mathematically extremely challenging and out of the scope of our current study.

\bibliographystyle{apalike}
\bibliography{Alvarez_Chirstensen_rev2019}

\end{document}